\documentclass[10pt]{article}
\usepackage[english]{babel}
\usepackage{amssymb,amsfonts,amsmath,amsthm}
\usepackage{fullpage}
\usepackage{xcolor}
\usepackage{graphicx}
\usepackage{authblk}
\usepackage{natbib}
\usepackage{subfigure,epstopdf}
\usepackage[title]{appendix}

\theoremstyle{plain}
\newtheorem{theorem}{Theorem}[section]
\newtheorem{corollary}[theorem]{Corollary}
\newtheorem{assumption}[theorem]{Assumption}
\newtheorem{lemma}[theorem]{Lemma}
\newtheorem{proposition}[theorem]{Proposition}

\theoremstyle{definition}
\newtheorem{definition}{Definition}

\theoremstyle{remark}
\newtheorem{remark}{Remark}

\newcommand{\cA}{\mathcal{A}}
\newcommand{\cB}{\mathcal{B}}
\newcommand{\cC}{\mathcal{C}}

\newcommand{\cF}{\mathcal{F}}

\newcommand{\cW}{\mathcal{W}}

\newcommand\abs[1]{\left|#1\right|}
\newcommand\Ep[1]{\mathbb{E}_{\P} \left[#1\right]}

\newcommand\normLip[1]{\|#1\|_{\rm Lip}}

\def \E{\mathbb{E}}

\def \N{\mathbb{N}}
\def \P{\mathbb{P}}
\def \Q{\mathbb{Q}}
\def \R{\mathbb{R}}

\def \X{\mathbb{X}}

\newcommand{\norm}[1]{\left\Vert#1\right\Vert}

\newcommand{\cNN}{\mathcal{NN}}
\newcommand{\fn}{r}
\newcommand{\normf}[1]{\norm{#1}_{\fn}}

\def \hX{\widehat{X}}

\newcommand{\hp}[2]{\widehat{\P}}
\def \hF{\widehat{F}}
\def \hP{\widehat{\P}}
\def \ind{\mathbf{1}}
\def \ide{\mathrm{Id}}
\def \ind{\mathbf{1}}

\def \lipone{\cA^0_{\rm Lip}}
\def \hphi{\widehat{\varphi}}
\def \argmin{\mathrm{argmin}}

\newcommand{\ES}{\mathop {\rm ES}\nolimits}
\newcommand{\entr}{\mathop {\rm Entr}\nolimits}

\newcommand{\trho}[2]{\widetilde{\rho}\left(#1|#2\right)}
\newcommand{\trhoone}[2]{\widetilde{\rho}_1\left(#1|#2\right)}
\newcommand{\trhotwo}[2]{\widetilde{\rho}_2\left(#1|#2\right)}
\usepackage{hyperref}

\begin{document}
\title{Risk Sharing with Deep Neural Networks}
\author[1]{M. Burzoni\thanks{
Email: matteo.burzoni@unimi.it.}}
\author[1]{A. Doldi\thanks{
Email: alessandro.doldi@unimi.it.}}
\author[2]{E. Monzio Compagnoni\thanks{
Email: enea.monziocompagnoni@unibas.ch.\\ All authors contributed equally to this paper.}}

\affil[1]{Universit\`a degli Studi di Milano}
\affil[2]{University of Basel}

\maketitle

\abstract{\noindent We consider the problem of optimally sharing a financial position among agents with potentially different reference risk measures. The problem is equivalent to computing the infimal convolution of the risk metrics and finding the so-called optimal allocations. We propose a neural network-based framework to solve the problem and we prove the convergence of the approximated inf-convolution, as well as the approximated optimal allocations, to the corresponding theoretical values. We support our findings with several numerical experiments.\\

}

\section{Introduction}
We consider the problem faced by $n$ economic agents, with reference risk measures $\rho_1,\ldots,\rho_n$, who want to share the risk carried by a certain financial position, represented by a random variable $X$. The goal is to write $X$ as the sum of $n$ random variables $X_1,\ldots,X_n$ so that the sum of the risk of the single agents, $\rho_1(X_1)+\cdots+\rho_n(X_n)$, is minimized. The problem is well known in the mathematical finance literature under the name of \emph{risk sharing}, and it amounts to the calculation of the infimal convolution (\emph{inf-convolution}) defined as follows:
\begin{equation}\label{eq:inf_conv_n}
		\rho_1\square\cdots\square \rho_n(X):=\inf\bigg\{\sum_{i=1}^n\rho_i(X_i)\ :\ \sum_{i=1}^nX_i=X\bigg\}.
\end{equation}
The seminal paper \citep{BE05}, which introduced inf-convolutions in the context of (convex) risk measures, originated a vast offspring of literature. \citep{Acciaio} and  \citep{FilipovicSvindland08} studied the case without monotonicity assumptions on the $(\rho_i)_{i=1,\ldots,n}$, \citep{MRG14} considered the case of cash-subadditive and quasi-convex functionals, while multivariate risks are treated in
 \citep{Carlier1,Carlier2}. We also mention 
\citep{Dana_LeVan},  \citep{Heath_Ku},  \citep{Tsanakas},
 \citep{Weber17}, \citep{LS18} and \citep{ELW18,ELW182} for further extensions and we refer to \citep{Ruschendorf13} for a comprehensive overview of the topic.

The most relevant results for our analysis were established by \citep{FilipovicSvindland08,JST07} in the study of the so-called \emph{optimal allocations} for $\rho_1\square\cdots\square \rho_n$, namely, the minimizers of the right-hand side of \eqref{eq:inf_conv_n}. For the case of law-invariant risk measures, it was demonstrated that comonotonicity plays a key role. In fact, optimal allocations can be found in the form $f_1(X),\ldots, f_n(X)$ for some non-decreasing, real-valued, maps $f_1,\ldots,f_n$, which sum up to the identity.
This key aspect inspires the numerical framework that we propose in this paper. Indeed, it can be shown that the functions $f_i$, for $i=1,\ldots,n$, are Lipschitz continuous functions, thus, they can be very well approximated by neural networks. Despite the abundance of theoretical results on the risk-sharing topic, we are not aware of a general framework for the numerical computation of the solutions which works under very little assumptions, such as law-invariance and convexity. Indeed, the aforementioned literature usually focused on the explicit (or semi-explicit) computation of the optimal allocations in some special cases. This operation obviously requires an exact computation of the inf-convolution and their minimizers, which needs to be worked out case-by-case. Some risk measures of interest, for which explicit solutions can be provided, are the \emph{entropic risk measure} and \emph{expected shortfall} among the family of convex risk measures and, more recently, the (Range) Value at Risk among the non-convex ones, see \citep{ELW18,ELW182}.   

Using a suitable version of the Universal Approximation Theorem, we prove in Section \ref{sec:framework} that
\[
		\rho_1\square \rho_2(X)=\inf\bigg\{\rho_1(f(X))+\rho_2(X-f(X)) : f\in {\cNN}\bigg\},
\]
where $\cNN$ is a suitable class of feed-forward neural networks. 
Deep neural networks (DNNs) have been proven to be very effective in solving a great variety of problems and in this paper, we show that this is the case also in a risk sharing context. The precise results are stated in Theorem \ref{thm:approx1} and in Theorem \ref{thm:approx2}, which constitute the main results of the section. Note that the restriction to $n=2$ agents is dictated by the convenience of exposition, but the case $n\geq2$ can be covered similarly. Of course, we do not exclude that other methods could be successfully applied. In Appendix \ref{app:AltMod}, we discuss some possible alternative approaches.

In view of establishing a rigorous framework for our numerical experiments, we devote Section \ref{sec:conv} to the convergence analysis of the historical estimators of the inf-convolution, as well as of their corresponding optimal allocations. Such estimators are constructed simply by applying the risk measures of the agents to the empirical distribution of a large sample of $X$ (see e.g. \citep{ContDeguestScandolo10} for an overview). The main convergence result of the section is Theorem \ref{thm:main_conv} which provides the theoretical justification of the experiments of Section \ref{sec:numerics}. We test our findings in a series of numerical experiments with different risk measures, different architectures, and different distributions for $X$ (see Section \ref{subs:algo} for the details about the framework) obtaining consistent results. As for the risk measures, we use the following: 
\begin{enumerate}
\item Entropic risk measure with parameter $\beta>0$: \begin{equation}\label{def:entropic}
    \entr_{\beta}(X):=\beta\log\E\big[{e^{- X/\beta}}\big];
\end{equation}
\item Expected Shortfall (ES) at level $\alpha\in(0,1)$:
\begin{equation*}
    \ES_\alpha(X)= \frac{1}{\alpha} \int_{0}^{\alpha}V@R_u(X)\mathrm{d}u,\qquad V@R_u(X):=\inf\{m\in\R \ : \ \P(X+m<0)\leq u\};
\end{equation*}
\item Distortion risk measure for $\mu\in \mathrm{Prob}([0,1])$, see \cite[Section 4.6]{FollmerSchied2}:
\begin{equation}\label{def:distRM}
\rho_\mu(X):=\int_0^1 \ES_\alpha(X) \mu(d\alpha).
\end{equation}
\end{enumerate}
The situation where all agents adopt an entropic risk measure, respectively ES, admits an explicit and simple solution for both the value of the inf-convolution and the optimal allocations. We test our numerical approximation in these cases only to confirm that the trained DNNs are converging to the known solutions. We then proceed in testing our algorithms in more complex situations. We cover the case of risk sharing between agents with distortion-type risk measures and between heterogeneous agents, that is, two agents using risk measures of different types --- one has the entropic and the other adopts either the expected shortfall or a distortion risk measure. In all such cases, we confirm that the trained network is able to recover the expected form of the optimal allocations, known from \citep{ELW18} and \citep{JST07}. In the last experiment, where we consider the convolution of an entropic risk measure with a distortion risk measure, we do not have any information about the solution.
 
We conclude this introduction with the frequently used notation.
For a metric space $\X$, $\cB(\X)$ denotes the Borel 
$\sigma $-algebra and $m\cB(\X)$ denotes the class of
real-valued, Borel-measurable functions on $\X$. We define the
following sets: 
\begin{align*}
	\mathrm{ca}(\X)& :=\left\{ \gamma :\cB(\X%
	)\rightarrow (-\infty ,+\infty )\ : \ \gamma \text{ is finite signed Borel
		measure on }\X\right\}; \\
	\mathrm{Meas}(\X)& :=\{\mu :\cB(\X)\rightarrow
	\lbrack 0,+\infty )\ : \ \mu \text{ is a non negative finite Borel measure on 
	}\X\}; \\
	\mathrm{Prob}(\X)& :=\{Q:\cB(\X)\rightarrow \lbrack
	0,1]\ : \ Q\text{ is a probability Borel measure on }\X\}; \\
	\cC(\X)& :=\{\varphi :\X\rightarrow {\mathbb{R}}\ : \
	\varphi \text{ is continuous on }\X\}; \\
	\cC_{b}(\X)& :=\{\varphi :\X\rightarrow {\mathbb{R}};
	\ : \ \varphi \text{ is bounded and continuous on }\X\};\\
	\mathrm{Prob}^p(\R)& :=\{\Q\in\mathrm{Prob}(\R) \ : \ \int_{\R}|x|^p\mathrm{d}\Q(x)<+\infty\},\,\,p\in[1,+\infty);\\
	\mathrm{Prob}^\infty_K(\R)& :=\{\Q\in\mathrm{Prob}(\R) \ : \ \Q([-K,K])=1\},\ K>0;\\
 \mathrm{Prob}^\infty(\R)& :=\bigcup_{K>0}\mathrm{Prob}^\infty_K(\R).
\end{align*}

\section{The theoretical framework}\label{sec:framework}
Let $(\Omega,\cF,\P)$ be  a standard non-atomic probability space (see e.g. \citep{Sv10} for details about the possibility of dropping the standardness assumption). The Banach space $L^p(\Omega,\cF,\P)$ for $p\in[1,\infty)$ is the set of $p$-integrable random variables on $(\Omega,\cF,\P)$, endowed with the norm $\|\cdot\|_p:=(\E[|\cdot|^p])^{1/p}$. The Banach space $L^\infty$ is the set of essentially bounded random variables, endowed with the supremum norm $\|\cdot\|_\infty$. The order relation $\le$ on such spaces is the one induced by the $\P$-a.s.\ ordering. We first recall the definition of monetary risk measures and some of their standard properties. We refer to the book \citep{FollmerSchied2} for a thorough presentation of the topic.
\begin{definition}
	\label{def:propertiesrm}Let $p\in[1,\infty]$ and $\rho:L^p(\Omega,\cF,\P)\to(-\infty,\infty]$ a functional. 
	\begin{itemize}
		\item $\rho$ is \emph{normalized} if $\rho(0)=0$;
		\item $\rho$ is \emph{finite} if $\rho(X)<\infty$ for every $X\in L^p(\Omega,\cF,\P)$;
		\item $\rho$ is \emph{monotone} if $\rho(X)\le\rho(Y)$, whenever $X\ge Y$ $\P$-a.s.;
		\item $\rho$ is \emph{cash additive} if $\rho(X+c)=\rho(X)-c$, for every $X\in L^p(\Omega,\cF,\P)$ and $c\in\R$;
		\item $\rho$ is \emph{convex} if $\rho(\lambda X+(1-\lambda)Y)\le \lambda\rho(X)+(1-\lambda)\rho(Y)$, for every $X,Y\in L^p(\Omega,\cF,\P)$ and $\lambda \in[0,1]$.
	\end{itemize}
	Any normalized, finite, monotone, and cash-additive $\rho$ is called a \emph{monetary risk measure}. If $\rho$ is also convex, it is called a \emph{convex risk measure}.
	\begin{itemize}
		\item $\rho$ is \emph{law-invariant} if $\rho(X)=\rho(Y)$ whenever $X\sim Y$;
		\item $\rho$ satisfies the \emph{Lebesgue Property} if $\rho(X)=\lim_{n\to+\infty}\rho(X_n)$ for any sequence $(X_n)_{n\in\N}\subseteq L^p(\Omega,\cF,\P)$ and $X\in L^p(\Omega,\cF,\P)$ such that: there exists $Z\in L^p(\Omega,\cF,\P)$ with $\abs{X_n}\leq Z\,\P$-a.s.\ for all $n\in\N$ and $\lim_{n\to+\infty}X_n=X\,\P$-a.s.\ holds.
	\end{itemize}
\end{definition}

We next introduce the concept of infimal convolution (inf-convolution in short) of convex risk measures $\rho_1\ldots,\rho_n$.
For ease of exposition, we restrict ourselves to the case of $n=2$, however, all the results generalize
to the case of an arbitrary $n\in\N$.
\begin{definition} Let $p\in[1,\infty]$. Given two functionals $\rho_1,\rho_2: L^p(\Omega,\cF,\P)\to(-\infty,\infty]$, their \emph{infimal convolution} is defined as:
	\begin{equation}\label{def:infconv}
		\rho_1\square \rho_2(X):=\inf\bigg\{\rho_1(X_1)+\rho_2(X_2)\ :\ X_1,X_2\in L^p,\  X_1+X_2=X\bigg\},\quad X\in L^p.
	\end{equation}
	Every couple $(X_1,X_2)\subseteq L^p$ such that $X_1+X_2=X$ is called an \emph{allocation} for $X$. Additionally, we say that an allocation is
	\begin{itemize}
		\item An \emph{optimal} allocation if it is a minimizer of the right-hand side of \eqref{def:infconv}; 
		\item A \emph{comonotonic} allocation if it is of the form $(f_1(X),f_2(X))$ for some increasing\footnote{Increasing is understood in the non-strict sense.} functions $f_1,f_2:\R\to\R$ such that $f_1+f_2=\ide$, where $\ide:\R\rightarrow\R$ denotes the identity function $\ide(x)=x$. 
	\end{itemize}
\end{definition}
The following well-known result, see  \citep{FilipovicSvindland08} Theorem 2.5, shows that for lower semi-continuous (l.s.c.) law-invariant convex risk measures, optimal allocations can be found among the class of comonotonic allocations. 
\begin{theorem}
	\label{thm:FS08}
	Let $p\in[1,\infty]$ and $\rho_1,\rho_2:L^p\rightarrow (-\infty,\infty]$ be l.s.c.\ law-invariant convex cash additive functions. Then $\rho_1\square \rho_2:L^p(\Omega,\cF,\P)\rightarrow [-\infty,\infty]$ is a l.s.c.\ law-invariant convex cash additive function. Moreover, there exist increasing functions $f_1,f_2:\R\rightarrow \R$ such that $f_1+f_2=\ide$ and 
	\[\rho_1\square \rho_2(X)=\rho_1(f_1(X))+\rho_2(f_2(X)).\]
\end{theorem}
It is not difficult to see that the functions $f_1, f_2$ are necessarily Lipschitz continuous. Indeed, for $x\ge y$, we can write $f_1(x)-f_1(y)+f_2(x)-f_2(y)=x-y$ and, using that both functions are increasing, we obtain the inequality $|f_i(x)-f_i(y)|\le |x-y|$ for $i=1,2$. For the case  $x\le y$, the argument is analogous. In particular, it holds $\normLip{f_i}\le 1$ for $i=1,2$, where 
\[\normLip{f}:=\inf\{L>0\ :\ |f(x)-f(y)|\le L |x-y|,\ \forall x,y\in\R\}.
\]
We also observe that, for monetary risk measures, if $(f_1(X),f_2(X))$ is an optimal allocation, the same is true for $(f_1(X)-c,f_2(X)+c)$ for an arbitrary $c\in\R$. This is also called \emph{rebalancing of cash}.  Without loss of generality, the function $f_1$ can be therefore chosen to satisfy $f_1(0)=0$, while still preserving the Lipschitz property.  
Combining these two observations we obtain the following corollary to Theorem \ref{thm:FS08}.
\begin{corollary}
	\label{cor:existoptimuminpipine}
	Under the assumptions of Theorem \ref{thm:FS08}, we have 
	\begin{equation}
		\label{eq:definfconv}
		\rho_1\square \rho_2(X)=\min\bigg\{\rho_1(f(X))+\rho_2(X-f(X))\ :\ f\in \lipone\bigg\},
	\end{equation}
	where 
	\begin{equation}
		\label{lip1}
		\lipone:=\left\{\ f:\R\to\R\ : f(0)=0,\ \normLip{f}\le 1,\ \normLip{\ide-f}\le 1
		\right\}
	\end{equation}
	is the set of \emph{normalized Lipschitz allocations}.
\end{corollary}
Any function $f\in\lipone$ induces the  allocation $(f(X), X-f(X))$. Indeed the sum equals $X$ by construction and, using the Lipschitz property, it is clear that if $X\in L^p(\Omega,\cF,\P)$ then $f(X)\in L^p(\Omega,\cF,\P)$ as well. With a slight abuse of terminology, we call \emph{allocation} also the pair of functions $(f,\ide-f)$. By denoting $\P_X$ the law of $X$ under $\P$, this terminology becomes accurate when we work on the probability space $(\R,\cB(\R),\P_X)$, as it will often be the case below. The following is a sufficient criterion for guaranteeing the uniqueness of optimal allocations, see \citep{FilipovicSvindland08} Proposition 3.1. For $X,Y\in L^p$ we use the notation $X-Y\notin\R$ for indicating that the difference $X-Y$ is not a constant random variable.
\begin{proposition}
	\label{prop:uniqueopt}
Under the assumptions of Theorem \ref{thm:FS08}, suppose additionally that $\rho_1$ is strictly convex in the sense that
	\begin{equation}
		\label{strictcnvxRM}
		\rho_1(\lambda Y+(1-\lambda) Z)<\lambda \rho_1(Y)+(1-\lambda)\rho_1(Z)\quad\forall \lambda \in (0,1), Y,Z\in L^p \text{ s.t. }Y-Z\notin\R.
	\end{equation}
	Then the optimal allocation for $\rho_1\square\rho_2$ is unique up to rebalancing of cash, namely, for any pair of optima $(X_1,X_2)$, $(\hX_1,\hX_2)\in L^p(\Omega,\cF,\P)\times L^p(\Omega,\cF,\P)$ it holds $\hX_i=X_i+c_i$ for some $c_i\in\R$ with $c_1=-c_2$ and for $i=1,2$. 
\end{proposition}

This proposition generalizes to the case $n\geq 2$ when all but one of the initial risk measures are strictly convex (see e.g. the discussion after Corollary 11.14 in \citep{Ruschendorf13}).
An example of strictly convex risk measure is the entropic risk measure of \eqref{def:entropic}. We observe here that, for a given $X\in L^p$, uniqueness can be obtained via a small perturbation of $\rho_1$ by guaranteeing, at the same time, that the value of the infimal convolution is close.
\begin{lemma}
	Let $\rho_1,\rho_2$ be law-invariant convex risk measures on $L^p(\Omega,\cF,\P)$, for $p\in[1,+\infty]$. Let $\tilde{\rho}$ be a strictly convex risk measure on $L^p(\Omega,\cF,\P)$. For every $\tilde{\varepsilon}>0$, the risk measure $\rho_{1,\tilde{\varepsilon}}:=(1-\tilde{\varepsilon})\rho_1+\tilde{\varepsilon}\tilde{\rho}$ is a strictly convex risk measure. Moreover, for every $X\in L^p(\Omega,\cF,\P)$ and $\varepsilon>0$, there exists $1>\tilde{\varepsilon}=\tilde{\varepsilon}(X)>0$ such that $|\rho_1\square\rho_2(X)-\rho_{1,\tilde{\varepsilon}}\square\rho_2(X)|<\varepsilon$.
\end{lemma}
\begin{proof}
	It is easy to see that all the properties of $\rho_1$ and $\tilde{\rho}$ are inherited by $\rho_{1,\tilde{\varepsilon}}$. 
 As for the second statement, in view of Corollary \ref{cor:existoptimuminpipine} it is enough to show that, defining $\Psi_{\tilde{\varepsilon}}(f):=\rho_{1,\tilde{\varepsilon}}(f(X))+\rho_2(X-f(X))$, $f\in \lipone
 $, we have $\lim_{\tilde{\varepsilon}\downarrow 0}\sup_{f\in\lipone}\abs{\Psi_{\tilde{\varepsilon}}(f)-\Psi_0(f)}=0$.
 To see this, observe that
 \[
 \abs{\Psi_{\tilde{\varepsilon}(f)}-\Psi_0(f)}=\tilde{\varepsilon}\abs{\rho_1(f(X))-\tilde{\rho}(f(X))}\leq\tilde{\varepsilon}\Big(\abs{\rho_1(f(X))}+\abs{\tilde{\rho}(f(X))}  \Big) 
\]
Since $f\in\lipone$, $\abs{x}\geq \abs{f(x)}\geq -\abs{x}$ and by monotonicity and finiteness of $\rho_1$ we get $\rho_1(-\abs{X})\geq \rho_1(f(X))\geq \rho_1(\abs{X})$, and $\abs{\rho_1(f(X))}\leq \abs{\rho_1(-\abs{X})}+\abs{\rho_1(\abs{X})}$. The same argument applies to $\tilde{\rho}$, from which we deduce $\abs{\Psi_{\tilde{\varepsilon}}(f)-\Psi_0(f)}\le \tilde{\varepsilon}K$ for some constant $K>0$ depending only on $X$.
Since the right-hand side does not depend on $f$, the claim is proved.
\end{proof}

Towards the aim of approximating the infimal convolutions using neural networks, we need some continuity of the risk functionals. For risk measures on $L^\infty$, the continuity is a consequence of the monotonicity and cash additivity properties. For the case $p\in[1,\infty)$, the Extended Namioka-Klee Theorem (see \citep{BF10}) guarantees that any proper convex and monotone functional on $L^p$ is continuous with respect to the $L^p$-norm, on the interior of its domain. Thanks to the finiteness property, convex risk measures as in Definition \ref{def:propertiesrm} are norm continuous for every $p\in[1,\infty]$ on the whole space. Throughout the paper, we will therefore make the following standing assumption. 
\begin{assumption}
	\label{ass:lawinv}
	$\rho_1$ and $\rho_2$ are law-invariant convex risk measures.
\end{assumption}

\subsection{Approximation of inf-convolutions via neural networks}
In this section, we show that the inf-convolution of two risk measures in \eqref{eq:definfconv} can be approximated using neural networks in the construction of the allocations. This is achieved by means of appropriate versions of the universal approximation theorem (UAT).  
We first note that we can reduce our focus to $(\R,\cB(\R))$.
Consider indeed a functional $\rho: L^p(\Omega,\cF,\P)\rightarrow (-\infty,+\infty]$ which is law-invariant. Since the underlying space is non-atomic, for every probability measure $\Q\in\mathrm{Prob}^p(\R)$ (or $\Q\in\mathrm{Prob}^\infty(\R)$ for $p=\infty$), there exists $X\in L^p(\Omega,\cF,\P)$ such that $\Q=\P_X$. Using the law invariance of $\rho$, the functional $\trho{\cdot}{\Q}:L^p(\R,\cB(\R),\Q)\rightarrow (\infty,+\infty]$ given by
\begin{equation}
\label{rhogivenP}
\trho{\varphi}{\Q}:=\rho(\varphi\circ X)\text{ for any measurable }X:\Omega\rightarrow \R\text{ such that }\P_X=\Q
\end{equation}
is well defined and it inherits the properties listed in Definition  \ref{def:propertiesrm} from $\rho$. A similar procedure has been considered by \citep{FrittelliMaggis18}, although with a totally different aim. We stress some key consequences.
\begin{proposition}
	\label{proptransferonR}
	Let $p\in [1,+\infty]$. Let $\rho_1,\rho_2: L^p\rightarrow\R$ be law-invariant convex risk measures. Then,
	\begin{align}
		\label{reform: convolution1}
		\phantom{=}\rho_1\square\rho_2(X)=&\inf\bigg\{\rho_1(f(X))+\rho_2(X-f(X))\ :\ f\in \lipone\bigg\}\\
		=&\inf\bigg\{\trhoone{f}{\P_X}+ \trhotwo{\ide-f}{\P_X}\ :\ f\in \lipone\bigg\}=\trhoone{\,\cdot\,}{\P_X}\square\trhotwo{\,\cdot\,}{\P_X}(\ide)\label{reform: convolution2}
	\end{align}
	and $\hat{f}\in \lipone$ is a minimum in \eqref{reform: convolution1} if and only if it is a minimum in \eqref{reform: convolution2}.
\end{proposition}
\begin{proof}
The equality \eqref{reform: convolution1} is simply Corollary \ref{cor:existoptimuminpipine}. The first equality in \eqref{reform: convolution2} is given by definition of $\tilde{\rho}$ in \eqref{rhogivenP} and the fact that any $f\in\lipone$ satisfies $f\in L^p(\R,\cB(\R),\P_X)$, thanks to the Lipschitz continuity.
The last equality in \eqref{reform: convolution2}  does not immediately follow from Corollary \ref{cor:existoptimuminpipine}, since we do not know if $(\R,\cB(\R),\P_X)$ is non-atomic. 
The inequality $\ge$ is clear. Using \eqref{rhogivenP}, we can rewrite
\begin{align*}
\trhoone{\,\cdot\,}{\P_X}\square\trhotwo{\,\cdot\,}{\P_X}(\ide)&=\inf\bigg\{\trhoone{Y}{\P_X}+\trhotwo{\ide-Y}{\P_X}\ :\ Y\in L^p(\R,\cB(\R),\P_X)\bigg\}\\
&=\inf\bigg\{\rho_1(Y\circ X)+\rho_2(X-Y\circ X)\ :\ Y\in L^p(\R,\cB(\R),\P_X)\bigg\}\\
&\geq \rho_1\square \rho_2(X)
\end{align*}
which concludes the proof of the equality chain in \eqref{reform: convolution1}, \eqref{reform: convolution2}. The last statement follows from \eqref{rhogivenP}.

\end{proof}

\begin{remark}
\label{remtoQ}The second equality in \eqref{reform: convolution2} holds, more generally, if we replace $\P_X$ with an arbitrary $\Q\in \mathrm{Prob}^p(\R)$. Indeed, since the space is non atomic, $\Q$ is the law of some $Y\in L^p(\Omega,\cF,\P)$. For such a $Y$ it also holds that $\rho_1\square\rho_2(Y)=\trhoone{\,\cdot\,}{\Q}\square\trhotwo{\,\cdot\,}{\Q}(\ide)$.
\end{remark}
We next introduce the class of neural networks that we intend to use.
\begin{definition}\label{def:NN} Let $L$, $N_0, \ldots, N_L\in \N$ with $L\ge 2$, let $\sigma: \R \to \R$ an activation function and for any $\ell=1,\ldots,L$, let $W_\ell: \R^{N_{\ell-1}}\to \R^{N_\ell}$ an affine function. A \emph{(feed-forward) neural network} is a function $F:\R^{N_0}\to \R^{N_L}$ defined as
	\[ F(x):=\left(W_L\circ \sigma\circ W_{L-1}\circ\cdots \sigma\circ W_1\right) (x),
	\]
where the activation function $\sigma$ is applied componentwise. 
\end{definition}
We denote by $\cNN$ the vector space generated by the class of neural networks from $\R^K\to\R$ determined by a fixed activation function $\sigma$, continuous, nonconstant and bounded. Notice that $\cNN$ is a vector subspace of $\cC_b(\R^K)$, and is, in particular, a convex cone. Moreover, imposing $L=2$ in the above definition, $\cNN$ contains the functions generated by only one hidden layer and only one output unit as considered in \citep{Hornik91}.
A simple argument based on the classical UAT of \citep{Hornik91} (Theorem 1) yields the approximation result that we need, at least for the case $p<+\infty$. The case $p=\infty$ is not covered by this theorem and we need a slightly different approach.
\begin{theorem}
	\label{thm:Hornik}
	Let $\sigma$ be continuous, bounded, and nonconstant. Then $\cNN$ is norm dense in $L^p(\R^K,\cB(\R^K),\mu)$ for any finite measure $\mu\in\mathrm{Meas}(\R^K)$ and $p\in [1,+\infty)$. 
\end{theorem}
The original theorem is stated for networks with two layers ($L=2$). Since $\cNN$ contains this particular class, the density result also holds as stated in Theorem \ref{thm:Hornik}.
The following is our first approximation result. 
\begin{theorem}\label{thm:approx1}
	Let $p\in[1,\infty)$. Let $\rho_1,\rho_2: L^p(\Omega,\cF,\P)\rightarrow\R$ be law-invariant convex risk measures. Then,
	\begin{equation}
		\label{eq:sup_on_nn} 
		\rho_1\square \rho_2(X)=\inf\bigg\{\rho_1(f(X))+\rho_2(X-f(X)) : f\in {\cNN}\bigg\}.
	\end{equation}
\end{theorem}
\begin{proof} Let $X\in L^p(\Omega,\cF,\P)$. Let $\P_X$ be the law of $X$ under $\P$. We claim that
	\begin{align*}
		\rho_1\square \rho_2(X)&=\inf\bigg\{\rho_1(f(X))+\rho_2(X-f(X))\ :\ f:\R\to\R\text{ s.t.\ }f(X)\in L^p(\Omega,\cF,\P)\bigg\}\\
		&=\inf\bigg\{\trhoone{f}{\P_X}+\trhotwo{\ide-f}{\P_X}\ :\ f\in L^p(\R,\cB(\R), \P_X)\bigg\}. 
	\end{align*}
	Indeed, from Corollary \ref{cor:existoptimuminpipine}, $\rho_1\square \rho_2(X)$ attains the minimum over the set of allocations $(f(X),X-f(X))$ with $f\in\lipone$. Since $f(X)\in L^p(\Omega,\cF,\P)$, for every $f\in\lipone$, the first equality follows. The second equality is by \eqref{rhogivenP}.
	Notice now that $\trhoone{\cdot}{\P_X}$ and $\trhotwo{\cdot}{\P_X}$ are law-invariant convex risk measures, which are $\|\cdot\|_p$-continuous by the Extended Namioka-Klee Theorem in \citep{BF10}. From Theorem \ref{thm:Hornik}, applied with $\mu=\P_X$, we have $$\rho_1\square \rho_2(X)=\inf\bigg\{\trhoone{f}{\P_X}+\trhotwo{\ide-f}{\P_X}\ :\ f\in \overline{\cNN}^{\norm{\cdot}_p}\bigg\},$$ which in turns yields, by norm continuity, that \eqref{eq:sup_on_nn} holds.
\end{proof}

The UAT does not provide uniform approximations and, in particular, it does not cover the $L^\infty$ case.
{\label{notnew}We use here an approach based on weighted spaces in order to obtain Theorem \ref{thm:UAT}, which is inspired by the forthcoming paper \citep{UAT_Teich}. This theorem is only instrumental for our main results and it is certainly not the first time that the theory of weighted spaces has been exploited in Universal Approximation results (see for example \citep{Kratsios21,UAT_Teich} and the references therein). We will also add a short proof for the sake of completeness.}
Let $\fn:\X\rightarrow [1,+\infty)$ be a continuous function with compact sublevels and define
\begin{equation*}
	C_{\fn}(\X):=\left\{ \phi \in \cC(\X)\ : \ \normf{\phi}:=\sup_{x\in \X}\frac{\left\vert \phi
		(x)\right\vert }{\fn (x) }<+\infty \right\} \,.
\end{equation*}
The space $C_{\fn}(\X)$ is a Banach lattice when endowed with the norm $\normf{\cdot}$. This can be easily verified following verbatim the argument for the standard case of bounded continuous functions with the supremum norm.
We introduce the following sets:
\begin{align*}
	\mathrm{ca}_\fn(\X)& :=\left\{\gamma :\cB(\X)\rightarrow
	\R\ : \ \mu \in \mathrm{ca}(\X)\text{ and } \int_\Omega \fn(x)\mathrm{d}\abs{\gamma}(x)<+\infty\right\}; \\\mathrm{Meas}_\fn(\X)& :=\left\{\mu :\cB(\X)\rightarrow
	\lbrack 0,+\infty )\ : \ \mu \in \mathrm{Meas}(\X)\text{ and } \int_\Omega \fn(x)\mathrm{d}\mu(x)<+\infty\right\}; \\
	\mathrm{Prob}_\fn(\X)& :=\left\{\Q:\cB(\X)\rightarrow \lbrack
	0,1]\ : \ \Q\in \mathrm{Prob}(\X)\text{ and }\int_\Omega \fn(x)\mathrm{d}\P(x)<+\infty\right\}; \\
	{B}_\fn(\X)& :=\overline{\cC_b(\X)}^{\normf{\cdot}}.
\end{align*}

\begin{proposition}\label{prop:isom}
	${B}_\fn(\X)$ is a Banach space and for every continuous linear functional $\ell\in (B_\fn(\X))^*$ there exists a unique $\gamma\in\mathrm{ca}_\fn(\X)$ such that $\ell (\phi)=\int_\X\phi\mathrm{d}\gamma$. Conversely, every $\gamma\in\mathrm{ca}_\fn(\X)$ defines a continuous linear functional in $ (B_\fn(\X))^*$ in the same way. 
\end{proposition}
\begin{proof}This follows from \citep{Dorsekteichmann22}, see Theorem 2.4 and Theorem 2.7.
\end{proof}

\begin{theorem}\label{thm:UAT}
	Let $\sigma:\R\rightarrow \R$ be continuous, bounded, and non-constant. Let $K\geq1$ be a fixed integer and take $\X=\R^K$. Then the family $\cNN$ in Definition \ref{def:NN} 
	is $\normf{\cdot}$-dense in $B_\fn(\R^K)$.
\end{theorem}
\begin{proof}
As commented above, it is enough to prove the result for networks with two layers.
	By \cite[Theorem 5]{Hornik91}, $\sigma$ is discriminatory, meaning that for any $\gamma\in \mathrm{ca}_\fn(\R^K)=B_\fn^*$ we have
	\[
	\int_{\R^K}\sigma\left(\sum_{j=1}^Ka_jx_j+\theta\right)\mathrm{d}\gamma=0,\forall a_1,\dots,a_K,\theta\in\R  \qquad \Rightarrow\qquad \gamma\equiv 0.
	\]
 
	Let $\overline{\cNN}$ be the weak closure of $\cNN$ in $B_\fn$ with respect to the topology $\sigma(B_\fn, B_\fn^*)=\sigma(B_\fn, \mathrm{ca}_\fn(\R^K))$, where the pairing is given by the integration and is well defined from Proposition \ref{prop:isom}. 
	Recall that, for a cone $C\subseteq B_\fn$, $C^\circ:=\{\gamma\in \mathrm{ca}_\fn(\R^K)\ :\  \int_{\R^K}\phi\mathrm{d}\gamma=0\,\forall\phi\in C\}$ is called the polar cone of $C$. Since $\sigma$ is discriminatory, $\overline{\cNN}^\circ\subseteq \cNN^\circ=\{0\}$. By the bipolar theorem, we have $\overline{\cNN}=\{0\}^\circ=B_\fn(\R^K)$. Since $\cNN$ in Definition \ref{def:NN} is convex, we have $\overline{\cNN}=\overline{\cNN}^{\normf{\cdot}}$, the latter being the $\normf{\cdot}$-closure of $\cNN$. This proves that $\overline{\cNN}^{\normf{\cdot}}=B_\fn(\R^K)$, as desired.
\end{proof}
The following is our second general approximation result.
\begin{theorem}\label{thm:approx2}
	Let  $\rho_1,\rho_2: L^\infty(\Omega,\cF,\P)\rightarrow\R$ be law-invariant convex risk measures. Then, 
	\begin{equation}
		\label{eq:sup_on_nn2} 
		\rho_1\square \rho_2(X)=\inf\bigg\{\rho_1(f(X))+\rho_2(X-f(X)) : f\in {\cNN}\bigg\},\qquad X\in L^\infty(\Omega,\cF,\P).
	\end{equation}
\end{theorem}
\begin{proof}
	Let $X\in L^\infty(\Omega,\cF,\P)$ and $\P_X$ its law under $\P$. From Proposition \ref{proptransferonR}, we can rewrite	
	\[
	\rho_1\square \rho_2(X)
	=\inf\bigg\{\trhoone{f}{\P_X}+\trhotwo{\ide-f}{\P_X}\ :\ f\in\lipone\bigg\}.
	\]
	Note that $\lipone\subseteq B^{\fn}(\R)$ for $\fn(x):=1+\abs{x}^{1+\varepsilon}$ with $\varepsilon>0$. We thus deduce,
	\begin{align*}
		\rho_1\square \rho_2(X)&=
		\inf\bigg\{\trhoone{f}{\P_X}+\trhotwo{\ide-f}{\P_X}\ :\ f\in\lipone\bigg\}\\
		&\ge \inf\bigg\{\trhoone{f}{\P_X}+\trhotwo{\ide-f}{\P_X}\ :\ f\in B^{\fn}(\R)\bigg\}\\
		&\ge \inf\bigg\{\trhoone{f}{\P_X}+\trhotwo{\ide-f}{\P_X}\ :\ f\in L^\infty(\R,\cB(\R), \P_X)\bigg\}\\
		&=\rho_1\square \rho_2(X),
	\end{align*}
	where in the second inequality we have used that $\P_X$ has compact support and $\fn$ is continuous, whereas, in the third inequality we used again the law invariance. Finally, the last equality is a consequence of Proposition \ref{proptransferonR}. This shows that all the above inequalities are actually equalities. From Theorem \ref{thm:UAT}, $B^{\fn}(\R)=\overline{\cNN}^{\normf{\cdot}}$, so that
	\[
	\rho_1\square \rho_2(X)=
	\inf\bigg\{\trhoone{f}{\P_X}+ \trhotwo{\ide-f}{\P_X}\ :\ f\in \overline{\cNN}^{\normf{\cdot}}\bigg\}.
	\]
	We conclude using the continuity of $\rho_{1},\rho_2$ with respect to the uniform convergence.
\end{proof}
\begin{remark}
	Note that if we replace $L^\infty(\Omega,\cF,\P)$ with $L^p(\Omega,\cF,\P)$ and $p\in[1,+\infty)$, the same arguments with the choice of $r(x):=1+|x|^{p+\varepsilon}$ yields that \eqref{eq:sup_on_nn2} holds for every $X\in L^{p+\varepsilon}(\Omega,\cF,\P)$. The proof presented for Theorem \ref{thm:approx1} is however more direct and it does not require extra integrability on $X$. The advantage of the approach of Theorem \ref{thm:approx2} is that the extra equality \[
	\rho_1\square \rho_2(X)=
	\inf\bigg\{\trhoone{f}{\P_X}+ \trhotwo{\ide-f}{\P_X}\ :\ f\in \overline{\cNN}^{\normf{\cdot}}\bigg\}
	\]
	shows that the elements in the closure are approximated by the neural networks uniformly also for the case of $p<\infty$. 
\end{remark}

\section{Convergence results}\label{sec:conv}
Let $X\in L^p(\Omega,\cF,\P)$, for some $p\in[1,+\infty]$. Consider an i.i.d.\ sequence $(X_n)_n\subseteq L^p(\Omega,\cF,\P)$ with common distribution $\P_X$ and let $(\hF_N)_N$ denote the corresponding sequence of empirical cumulative distribution functions: for $x\in\R$, $\hF_N(x):\Omega\rightarrow \R$ is defined as $$\hF_N(x):=\frac1N\sum_{n=1}^N\ind_{(-\infty,x]}(X_n),\,\,x\in\R.$$ Denote by $\hP_N$ the random measure associated to the empirical c.d.f.\ $\hF_N$, namely 
\begin{equation}\label{Pemp}
\omega\mapsto \hP_N(\omega)(\cdot):=\frac{1}{N}\sum_{n=1}^N\delta_{X_n(\omega)}(\cdot).
\end{equation}
Finally, for $p\in[1,+\infty)$, let $\cW_p$ be the $p$-Wasserstein distance on $\mathrm{Prob}^p(\R)$ induced by the Euclidean norm, namely,
\[
\cW_p(\mu,\nu)=\inf\left\{(\E|X-Y|^p)^{\frac{1}{p}}\ :\ X\sim\mu,\ Y\sim \nu \right\}.
\]
We refer to the book \citep{Villani09} for a thorough presentation of the topic.
\begin{lemma} 
	\label{LLNforlaws}
	Let $p\in[1,+\infty]$ and $X\in L^p(\Omega,\cF,\P)$. Then $\hP_N\Rightarrow_N\P_X$, $\P$-a.s., where $\Rightarrow$ denotes the weak convergence of probability measures.
	If $p\in[1,+\infty)$, it holds additionally
	$$\lim_{N\rightarrow+\infty}{W_p(\P,\hP_N)}=0\,\, \P\text{-a.s.}.$$ 
\end{lemma} 
\begin{proof}
	The first statement follows from the Glivenko-Cantelli theorem. As for the second statement, Theorem 6.9 in \citep{Villani09} shows that convergence in $\cW_p$ is equivalent to $\hP_N\Rightarrow_N\P_X$ together with the convergence of the $p$-th moments. The latter follows from the law of large numbers and the integrability of $\P_X$, so that, $\lim_{N\rightarrow +\infty}\int_{\R}|x|^pd\hP_N(x)=\int_{\R}|x|^pd\P_X(x)$.
\end{proof}

We aim at proving the convergence of the optimal values, namely,
\begin{equation}
	\label{eq:llnrisksharing}
	\rho_1\square \rho_2(X)=\lim_{N\rightarrow +\infty}\trhoone{\cdot}{\hP_N(\omega)}\square\trhotwo{\cdot}{\hP_N(\omega)}(\ide),\qquad \P\text{-a.e. }\omega,
\end{equation}
and the convergence of the corresponding minimizers. We will need to establish some joint continuity results for the map $(f,\Q)\mapsto\trho{f}{\Q}$, defined in \eqref{rhogivenP}, on the spaces $\lipone \times \mathrm{Prob}^{p}(\R)$ and $\lipone \times \mathrm{Prob}^\infty_K(\R)$. Some preliminary topological considerations are useful. 
\begin{remark}
	\label{remboundforempirical}The space $\mathrm{Prob}^\infty_K(\R)$ is used for the $L^\infty(\Omega,\cF,\P)$ case. Indeed, for $X\in L^\infty(\Omega,\cF,\P)$ and $(X_n)_n$ an i.i.d.\ sequence with common law $\P_X$, we obviously have $|X_n|\le \norm{X}_{\infty}$ $\P$-a.s.. Thus, the measure $\hP_N:=\frac{1}{N}\sum_{n=1}^N\delta_{X_n}$ satisfies $\hP_N(\omega)\in \mathrm{Prob}^\infty_K(\R)$ for $\P$-a.e.\ $\omega$ and for $K\ge\norm{X}_{\infty}$.
\end{remark}
We endow $\mathrm{Prob}^{p}(\R)$ with the $\cW_p$-topology, for $p\in[1,\infty)$ and $\mathrm{Prob}^\infty_K$ with the weak topology. We endow $\lipone$ with the topology induced by the metric:
\begin{equation}
	\label{normcpt}
	d(\varphi,\psi):=\sum_{h=1}^\infty\frac{1}{2^h}\min\left(1,\sup_{x\in[-h,h]}\abs{\varphi(x)-\psi(x)}\right),\qquad \varphi,\psi\in\lipone.\\
\end{equation}
One can verify that $(\cC(\R),d)$ is a complete metric space and $d$ metrizes the uniform convergence on compact sets.
As a consequence of the Ascoli-Arzel\`a theorem, we have that $\lipone$ is compact with respect to the topology induced by $d$, as we prove next.
\begin{lemma}
	\label{lem:ascoliarze}
	Let $(\varphi_n)_n$ be a sequence in $\lipone$.	There exists $\varphi\in \lipone$ and a subsequence $(\varphi_{n_k})_k$ such that $\lim_{n\rightarrow+\infty}d(\varphi_{n_k},\varphi)=0$.
\end{lemma}
\begin{proof}
	Recall that for every $f\in\lipone$, $\normLip{f}\le 1$. In particular, for every $h\in\N$, any family of functions in $\lipone$, restricted to $[-h,h]$, is equicontinuous and equibounded. We first construct a sequence of functions $(\varphi_h)_{h\in\N}$ in the following iterative way. For $h=1$ we apply the Ascoli-Arzel\`a theorem to $(\varphi_n)_n$ restricted to $[-1,1]$. This yields a subsequence, that we relabel again as $(\varphi_n)_n$, and a continuous function $\varphi_1$ on $[-1,1]$ such that $(\varphi_n)_n$ converges uniformly to $\varphi_1$ on $[-1,1]$. Note that $\varphi_1$ is continuous on $[-1,1]$, being uniform limit of continuous functions, and satisfies $\varphi_1(0)=0$, since $\normLip{\varphi_n}\le 1$ and $\varphi_n(0)=0$ for any $n\in\N$. 
	At the step $h+1$ we repeat the same argument to the sequence $(\varphi_n)_n$ obtained at step $h$. Note that the limiting function $\varphi_{h+1}$ satisfies $\varphi_{h+1}=\varphi_{h}$ on $[-h,h]$, since $(\varphi_n)_n$ converges uniformly to $\varphi_{h}$ on $[-h,h]$. Similarly as above, $\varphi_{h+1}$ is continuous on $[-(h+1),h+1]$ and satisfies $\varphi_{h+1}(0)=0$.
	
	We are now able to construct the limiting $\varphi\in\lipone$. For every $h\in\N$, we extend $\varphi_h$ to $\R$ in an arbitrary way outside $[-h,h]$. We set $\varphi(x):=\lim_{h\to\infty}\varphi_{h}(x)$, for every $x\in\R$ and note that $\varphi$ coincides with $\varphi_{h}$ on every $[-h,h]$. In particular, we deduce that $\varphi\in\lipone$.
	
	Finally, we construct the convergent subsequence $(\varphi_{n_k})_{k}$ of the original sequence. From the Ascoli-Arzel\`a argument above, for every, $k\in\N$, there exists $n_k\in\N$ such that \[
	\sup_{x\in[-k,k]}\abs{\varphi_{n_k}(x)-\varphi_k(x)}<\frac{1}{k}.\]
	For every $h\in \N$ and $k\ge h$, using that $\varphi=\varphi_{k}$ on $[-k,k]$, we obtain
	\[
	\sup_{x\in[-h,h]}\abs{\varphi_{n_k}(x)-\varphi(x)}\le \sup_{x\in[-k,k]}\abs{\varphi_{n_k}(x)-\varphi(x)}= \sup_{x\in[-k,k]}\abs{\varphi_{n_k}(x)-\varphi_k(x)}<\frac{1}{k},
	\]
	which implies the uniform convergence of the subsequence  $(\varphi_{n_k})_{k}$ to $\varphi$ on every interval $[-h,h]$. An application of Dominated Convergence Theorem then yields $$\lim_{k\to +\infty}d(\varphi_{n_k},\varphi)=\lim_{k\to +\infty}\sum_{h=1}^\infty\frac{1}{2^h}\min\left(1,\sup_{x\in[-h,h]}\abs{\varphi_{n_k}(x)-\varphi(x)}\right)=0.$$
\end{proof}
We state now our main convergence result. We need the following functional, which is an almost sure version of the distance $d$. For $\mu$ a measure on $\cB(\R)$ and $\varphi,\psi:\R\rightarrow\R$, $\cB(\R)$-measurable functions,
\begin{equation}
	\label{dmu}
	d_\mu(\varphi,\psi):=\sum_{h=1}^\infty\frac{1}{2^h}\norm{\min\left(1,\abs{\varphi-\psi}1_{[-h,h]}\right)}_{\infty,\mu},
\end{equation}
where in the above $L^\infty(\R,\cB(\R),\mu)$-norm we made explicit the dependence on the reference measure $\mu$ in order to avoid ambiguity in what follows.

\begin{theorem} \label{thm:main_conv}
	Let $p\in[1,\infty]$. Let $\rho_1,\rho_2: L^p(\Omega,\cF,\P)\rightarrow\R$ be law-invariant convex risk measures. Suppose that $\rho_1$ is strictly convex in the sense of \eqref{strictcnvxRM}. Only in the case of $p=\infty$ suppose that $\rho_1$ and $\rho_2$ satisfy the Lebesgue property.
	Then, 
	\begin{equation*}
		\trhoone{\cdot}{\P_X}\square\trhotwo{\cdot}{\P_X}(\ide)=\lim_{N\rightarrow +\infty}\trhoone{\cdot}{\hP_N(\omega)}\square\trhotwo{\cdot}{\hP_N(\omega)}(\ide)\qquad \P\text{-a.e. }\omega.
	\end{equation*}
	Furthermore, let $(\hphi,\ide-\hphi)$ and $(\hphi_N(\omega),\ide-\hphi_N(\omega))$ be  optimal allocations in $\lipone$, corresponding to $\P_X$ and $\hP_N(\omega)$ respectively. Then, 
	\[\lim_{N\to +\infty}d_{\P_X}(\hphi_N(\omega),\hphi)=0\,\qquad \P\text{-a.e. }\omega.\]	
\end{theorem}
The rest of the section is devoted to the proof of this theorem. We will need a number of auxiliary results, which are of independent interest. The first result is essentially \cite[Proposition 1]{Delbaen21} or \cite[Theorem 2.1]{Shapiro13}, adapted to our context.
\begin{lemma}
	\label{lemmadelbaen}
	Consider again a generic atomless probability space $(\Omega,\cF,\P)$. 
	\begin{enumerate}
		\item \label{item1}Let $p\in[1,+\infty)$. Suppose $(\Q_n)_n,\Q\subseteq\mathrm{Prob}^p(\R)$ and $\lim_{n\to+\infty}W_p(\Q_n,\Q)=0$. Then, there exists a sequence $(Y_n)_n$ in $L^p(\Omega,\cF,\P)$ and $Y\in L^p(\Omega,\cF,\P)$ such that $\P_{Y_n}=\Q_n$ for every $n\in\N$, $\P_Y=\Q$ and $\lim_{n\to+\infty}\norm{Y_n-Y}_p=0$.
		\item \label{item2}Let $p=\infty$. Suppose $(\Q_n)_n,\Q\subseteq\mathrm{Prob}_K^\infty(\R)$ and $\Q_n\Rightarrow_n\Q$. Then there exists a sequence $(Y_n)_n$ in $L^\infty(\Omega,\cF,\P)$ and $Y\in L^\infty(\Omega,\cF,\P)$ such that $\P_{Y_n}=\Q_n$ for every $n\in\N$, $\P_Y=\Q$, $\lim_{n\to+\infty }Y_n=Y\,\P$-a.s. and $\sup_{n}\norm{Y_n}_\infty<+\infty$.
	\end{enumerate}
\end{lemma}
\begin{proof} For every $\Q_n\Rightarrow_n\Q$, by the Skorokhod Theorem (as in \cite[Theorem 25.6]{Billingsley}), there exist random variables  $Y,(Y_n)_n$ such that $\P_Y=\Q$, $\P_{Y_n}=\Q_n$ for every $n\in\N$ and $\lim_{n\to+\infty}Y_n=Y\,\P$-a.s. For item \ref{item2}, it is enough to additionally note that $\P(Y_n\in[-K,K])=\Q_n([-K,K])=1$.

As for item \ref{item1}, by the characterization of $W_p$-convergence in \cite[Theorem 6.9]{Villani09}, we have $$\lim_{n\to+\infty}\Ep{\abs{Y_n}^p}=\lim_{n\to+\infty}\int_{\R}\abs{x}^p\mathrm{d}\Q_n(x)=\int_{\R}\abs{x}^p\mathrm{d}\Q(x)=\Ep{\abs{Y}^p}.$$ Now we proceed as in \citep{Delbaen21} Proposition 1: by Scheffé's lemma we conclude that $(\abs{Y_n}^p)_n$ converges in $L^1(\Omega,\cF,\P)$ to $\abs{Y}^p$. Hence $\abs{Y_n-Y}^p\leq \frac{1}{2}\abs{2Y_n}^p+\frac{1}{2}{\abs{2Y}^p}=2^{p-1}(\abs{Y_n}^p+\abs{Y}^p)$ is a uniformly integrable sequence. Indeed,
\begin{align*}
 &\Ep{\ind_{\{\abs{Y_n-Y}^p\geq K\}}\abs{Y_n-Y}^p}\leq \Ep{\ind_{\{\abs{Y_n}^p+\abs{Y}^p\geq 2^{1-p}K\}}\abs{q_n-q}^p}\\
 &\leq \Ep{\ind_{\{\abs{Y_n}^p+\abs{Y}^p\geq 2^{1-p}K\}}\left(\abs{Y_n}^p+\abs{Y}^p\right)}.   
\end{align*}
	Since $\abs{Y_n-Y}^p$ converges to zero $\P$-a.s., the proof is complete.
\end{proof}
\begin{proposition}
	\label{prop:trhocont}Let $p\in[1,\infty]$ and $\rho: L^p(\Omega,\cF,\P)\rightarrow\R$ be a law-invariant convex risk measure. For $p\in [1,+\infty)$ the map
	$(f,\Q)\mapsto\trho{f}{\Q}$ of \eqref{rhogivenP} is continuous on $\lipone \times \mathrm{Prob}^{p}(\R)$. If, on the other hand $p=+\infty$ and $\rho$  additionally satisfies the Lebesgue property on $L^\infty(\Omega,\cF,\P)$, then  $(f,\Q)\mapsto\trho{f}{\Q}$ is continuous on $\lipone \times\mathrm{Prob}^{\infty}_K(\R)$, for every $K>0$.
\end{proposition}
\begin{proof}
	We start covering the case $p\in[1,+\infty)$.
	Since $\lipone\times \mathrm{Prob}^p(\R)$ is a metric space, we check continuity along sequences.
	Take a convergent sequence $(f_n,\Q_n)\rightarrow_n(f,\Q)$ in $\lipone \times \mathrm{Prob}^{p}(\R)$. Take any subsequence. We prove that it admits a further  subsequence for which $\lim_{k\to+\infty}\trho{f_{n_k}}{\Q_{n_k}}=\trho{f}{\Q}$, which yields the convergence of the original sequence. The first extracted subsequence will be relabelled with the index $n\in\N$. Since $\lim_{n\to+\infty}W_p(\Q_n,\Q)=0$, we can  apply Lemma \ref{lemmadelbaen}. Since $$\lim_{N\to+\infty}\norm{Y_n-Y}_p=0$$ up to taking a further subsequence (and relabeling again with $n$) we might suppose that there exists a $0\leq Z\in L^p$ with $\abs{Y_n}\leq Z\,\forall n\in\N$, $\P$-a.s.
	By Dominated Convergence Theorem, since $\lim_{n\to+\infty}f_n(Y_n)=f(Y)\,\P$-a.s., we get $$\lim_{n\to+\infty}\norm{f_n(Y_n)-f(Y)}_p=0.$$
	Now, since $\rho$ is real-valued, hence norm continuous by Extended Namioka-Klee Theorem, we deduce
	\begin{equation}
		\label{eq:provescont}
		\trho{f}{\Q}:=\rho(f(Y))=\lim_{n\to+\infty}\rho(f_n(Y_n))=:\lim_{n\rightarrow +\infty}\trho{f_n}{\Q_n}
	\end{equation}
	and the desired continuity follows. 
	
	For the second statement, the argument is very similar. Take a convergent sequence $(f_n,\Q_n)\rightarrow_n(f,\Q)$ in $\lipone \times \mathrm{Prob}^{\infty}_K(\R)$. Take any subsequence. We prove that it admits a further subsequence for which $\lim_{k\to+\infty}\trho{f_{n_k}}{\Q_{n_k}}=\trho{f}{\Q}$. Use Lemma \ref{lemmadelbaen} item \ref{item2} to obtain the sequence $(Y_n)_n$. Since $f_n\in \lipone\,\forall n\in\N$, we have $\sup_n\norm{f_n(Y_n)}_{\infty}\leq \sup_n\norm{Y_n}_{\infty}$ which is finite by Lemma \ref{lemmadelbaen} item \ref{item2}. Since $f_n$ converges to $f$ uniformly on compact intervals by definition, we deduce $\lim_{n\to+\infty} f_n(Y_n)=f(Y)\,\P$-a.s. Using the Lebesgue property we conclude that \eqref{eq:provescont} holds true, providing continuity.
\end{proof}

\begin{proof}[Proof of Theorem \ref{thm:main_conv}]Consider first the case $p\in[1,\infty)$. We prove something stronger, namely, that the thesis holds for every $(\Q_n)_n,\Q\subseteq\mathrm{Prob}^p(\R)$ such that $\lim_{n\to+\infty}W_p(\Q_n,\Q)=0$ instead of only for $(\hP_N)_N$ and $\P_X$. 
	First, observe that by Proposition \ref{prop:trhocont} the function $$(\varphi,\Q)\mapsto \trhoone{\varphi}{\Q}+\trhotwo{\ide-\varphi}{\Q}$$ is continuous on $\lipone\times \mathrm{Prob}^p(\R)$, and $\lipone$ is compact by Lemma \ref{lem:ascoliarze}. Berge's Theorem (see \citep{Aliprantis} Theorem 17.31) guarantees that  \begin{equation*}
		\trhoone{\cdot}{\Q}\square\trhotwo{\cdot}{\Q}(\ide)=\lim_{N\rightarrow +\infty}\trhoone{\cdot}{\Q_n}\square\trhotwo{\cdot}{\Q_n}(\ide)
	\end{equation*} 
	and that the correspondence $\Gamma:\mathrm{Prob}^p(\R)\rightrightarrows\lipone$ defined by $$\Gamma(\Q):=\argmin\left\{\trhoone{\varphi}{\Q}+\trhotwo{\ide-\varphi}{\Q}\ : \ \varphi\in\lipone\right\}$$
	is upper hemicontinuous. 
	Consider now the numerical sequence $(d_{\Q}(\hphi_n,\hphi))_n$. Take an arbitrary subsequence and relabeled it again by $n$. Using the upper hemicontinuity of $\Gamma$ (see \citep{Aliprantis} Theorem 17.20) and the convergence of $(\Q_n)_n$ to $\Q$, the sequence  $(\hphi_n\in\Gamma(\Q_n))_n\subseteq\lipone$  has a limit point in $\Gamma(\Q)$, that we call $\hphi_\infty$. Up to passing to a further subsequence and relabeling, we may assume that $(\hphi_n)_n$ converges to $\hphi_\infty$ with respect to the distance $d$. By definition of $\Gamma$, $\hphi_\infty$ induces an optimal allocation under $\Q$ and for $Y\in L^p(\Omega,\cF,\P)$ with $\P_Y=\Q$ we get, and using Proposition \ref{proptransferonR},
	$$\rho_1\square \rho_2(Y)=\rho_1(\hphi_\infty(Y))+\rho_2(Y-\hphi_\infty(Y)).$$ 
	Since $\rho_1$ is strictly convex, the minimizer is unique by Proposition \ref{prop:uniqueopt} (recall that we fixed $\hphi_\infty(0)=0$\footnote{{As a consequence of translation invariance, we can assume, without loss of generality, that  $\Q$ gives positive mass to every neighborhood of $0$. Then, given two optimal allocations $\varphi_1,\varphi_2\in\lipone$, Proposition \ref{prop:uniqueopt} implies that $\varphi_1-\varphi_2=c$ $\Q$-a.s., for some $c\in\R$. However, the definition of $\lipone$ necessarily implies $c=0$.}}). Thus, $\P(\hphi(Y)\neq\hphi_\infty(Y))=0$, or equivalently $0=\P_Y(\hphi\neq\hphi_\infty)=\Q(\hphi\neq\hphi_\infty)$. The latter $\Q$-a.s.\ equality property yields $d_\Q(\hphi_n,\hphi_\infty)=d_\Q(\hphi_n,\hphi)$. Note now that, by construction, $d_\Q\le d$. We conclude that  $$\limsup_{n\to+\infty}d_\Q(\hphi_n,\hphi)=\limsup_{n\to+\infty}d_\Q(\hphi_n,\hphi_\infty)\leq 
	\lim_{n\to+\infty}d(\hphi_n,\hphi_\infty)=0.$$ We have shown that starting from an arbitrary subsequence of $(d_{\Q}(\hphi_n,\hphi_\infty))_n$ there exists a further subsequence converging to $0$. This shows the desired property.
	For the case $p=+\infty$, the argument is exactly the same. Note only that for applying Proposition \ref{prop:trhocont} we need to require the Lebesgue continuity. 
 
 The claims in the statement follow now from Lemma \ref{LLNforlaws}.
\end{proof}

\paragraph{The case of spectral risk measures} The convergence \eqref{eq:llnrisksharing} can be established in the context of spectral risk measures by proving a stronger result.

\begin{definition}
	\label{def:spectral}
	Let  $p\in[1,\infty]$. A functional $\rho: L^p\rightarrow (-\infty,\infty]$ is called a \emph{spectral risk measure}  if 
	\begin{equation}
		\label{eq:spectral}
		\rho(X)=\int_{0}^1 V@R_\alpha(X)h(\alpha)\mathrm{d}\alpha,
	\end{equation}
	for some non-increasing function $h:[0,1]\to[0,\infty)$, called \emph{spectral density}, satisfying $\int_0^1 h(p)\mathrm{d}p=1$.
\end{definition}
We refer to \citep{Pichler13,Pichler13b} for a thorough analysis of the topic. In particular, the properties of $h$ ensure that $\rho$ is convex, monotone, and cash additive. Moreover, due to the properties of $ V@R$, $\rho$ is also law invariant and positive homogeneous. Whenever $\rho_1$ and $\rho_2$ are both finite spectral risk measures on $L^p(\Omega,\cF,\P)$, $p\in [1,+\infty]$ (which is the case for suitably integrable spectral densities as shown in \cite[Proposition 5 and Theorem 11]{Pichler13b}), Assumption \ref{ass:lawinv} is thus satisfied. For $p=+\infty$, $\rho_1$ and $\rho_2$ also satisfy the Lebesgue property if the spectral densities $h_1,h_2$ are such that $\rho_1$ and $\rho_2$ are well defined and finite on $L^r(\Omega,\cF,\P)$ for some $r$ big enough, by the Extended Namioka-Klee Theorem. This translates into an integrability requirement on $h_1,h_2$, see \cite[Proposition 5 and Theorem 11]{Pichler13b}, and holds true for example if $h_1,h_2$ are bounded themselves. In both cases, we are in the exact setup of Theorem \ref{thm:main_conv}. However, we can provide an explicit estimate for the convergence \eqref{eq:llnrisksharing}, as detailed below.

\begin{proposition}Let $p\in (1,\infty)$ and $\rho_1,\rho_2: L^p(\Omega,\cF,\P)\rightarrow\R$ be law-invariant spectral convex risk measures, with possibly different spectral densities $h_1, h_2\in L^{\frac{p}{p-1}}([0,1],\cB([0,1]),\mathrm{Leb})$. Then, 
	\begin{equation}
		\label{growthcontrol}
		\abs{\rho_1\square\rho_2(X)-\trhoone{\cdot}{\hP_N(\omega)}\square \trhotwo{\cdot}{\hP_N(\omega)}(\ide)}\leq \left[\norm{h_1}_{\frac{p}{p-1}}+\norm{h_2}_{\frac{p}{p-1}}\right]W_p(\hP_N(\omega),\P_X)\,\,\forall\,\omega\in\Omega.
	\end{equation}
	In particular, \eqref{eq:llnrisksharing} holds.
\end{proposition}
\begin{proof}
	We see that, fixing  $\omega \in \Omega$ and taking $\hP_N(\omega)$ as a (deterministic) measure in $\mathrm{Prob}^p(\R)$, we also have $\P_X\in \mathrm{Prob}^p(\R)$ since $X\in L^p(\Omega, \cF,\P)$, and 
	\begin{align*}
		\rho_1\square\rho_2(X)&\stackrel{\text{Prop.} \ref{proptransferonR}}{=}\inf\bigg\{\trhoone{f}{\P_X}\square \trhotwo{\ide-f}{\P_X}\ :\ f\in \lipone\bigg\}\\
		&\leq\inf\bigg\{\trhoone{f}{\hP_N(\omega)}+ \trhotwo{\ide-f}{\hP_N}\ :\ f\in \lipone\bigg\}\\
		&+ W_p(\P_X,\hP_N)\norm{h_1}_{\frac{p}{p-1}}+ W_p(\P_X,\hP_N)\norm{h_2}_{\frac{p}{p-1}}\\
		&\stackrel{\text{Prop.} \ref{proptransferonR}}{=}\trhoone{\cdot}{\hP_N}\square \trhotwo{\cdot}{\hP_N(\omega)}(\ide)+\left[\norm{h_1}_{\frac{p}{p-1}}+\norm{h_2}_{\frac{p}{p-1}}\right]W_p(\P_X,\hP_N(\omega))
	\end{align*}
 where the inequality follows from \cite[Corollary 11]{Pichler13}.
	Interchanging the roles of $\P_X,\hP_N(\omega)$, we get 
	$$\trhoone{\cdot}{\hP_N(\omega)}\square \trhotwo{\cdot}{\hP_N(\omega)}(\ide)\leq \rho_1\square\rho_2(X)+\left[\norm{h_1}_{\frac{p}{p-1}}+\norm{h_2}_{\frac{p}{p-1}}\right]W_p(\hP_N(\omega),\P_X)$$
	so that \eqref{growthcontrol} holds.
	By Lemma \ref{LLNforlaws} there exists $E\in\cF$ with $\P(E)=0$ such that $\lim_{N\rightarrow+\infty}W_p(\P_X,\hP_N(\omega))=0$ for all $\omega\in \Omega\setminus E$. Thus, by \eqref{growthcontrol}, we have  $$\rho_1\square \rho_2(X)=\lim_{N\rightarrow +\infty}\trhoone{\cdot}{\hP_N}\square\trhotwo{\cdot}{\hP_N}(\ide)\,\,\forall\omega\in \Omega\setminus E\,.$$
	
\end{proof}

\section{Numerical experiments}\label{sec:numerics}
In this section, we illustrate the results of a number of numerical experiments that showcase the usefulness of the approximation developed in Section \ref{sec:framework}. We first test our findings in the case of entropic risk measures and expected shortfall, where simple explicit formulas for the optimal allocations and for the value of the inf-convolutions are known. We then consider the more complex case of distortion risk measures and, to conclude, we treat the case of heterogeneous agents adopting risk measures in two different classes, i.e. entropic and distortion-type.

\subsection{Description of the framework}\label{subs:algo}
We model two agents with reference risk measures $\rho_1$ and $\rho_2$ as those in the introduction. For the sake of comparison, we consider a financial position $X\in L^\infty(\Omega,\cF,\P)$, as it belongs to the domain of each of those risk measures. The objective is to approximate optimal allocations for the inf-convolution of $\rho_1$ and $\rho_2$, which takes the form 
\[\rho_1\square \rho_2(X)=\inf\bigg\{\trhoone{f}{\P_X}+ \trhotwo{\ide-f}{\P_X}\ :\ f\in \lipone\bigg\}.\]

To model the functions $f$ and $\ide-f$, we use two Fully Connected Deep Neural Networks (DNNs) $\phi_1$ and $\phi_2$, respectively. We observe that, while $\phi_1$ and  $\phi_2$ explicitly parametrize $f$ and $\ide-f$ by design, the functions $\ide-\phi_1$ and $\ide-\phi_2$ are proxies for $\ide-f$ and $f$ respectively.

Let $\Tilde{X} = \left( X_1, \cdots, \ X_N \right)$ be a sample of $\P_X$ of size $N$ with $\hP_N$ its empirical measure. For $i=1,2$, we denote by $\widehat{\rho_i}$ the historical risk measures associated to $\rho_1$ and $\rho_2$, namely, $\widehat{\rho_i}(\cdot):=\tilde{\rho}_i(\cdot\mid\hP_N)$. Since we aim at finding allocations that realize an inf-convolution value as close as possible to $\rho_1\square \rho_2({X})$, we need to find an appropriate and robust estimate of such a quantity. While we could  use the explicit parametrizations $\phi_1$ of $f$ and $\phi_2$ of $\ide - f$ and minimize $\hat{\rho_1}(\phi_1(X))+ \hat{\rho_2}(\phi_2(X))$, another valid alternative is to use the implicit parametrizations and minimize $\hat{\rho_1}(\ide - \phi_2(X))+ \hat{\rho_2}(\ide-\phi_1(X))$. To have a more robust estimate, we use their arithmetic average which leads to the following loss function
\begin{equation}\label{eq:loss}
L_{\rho_1,\rho_2}(\Tilde{X}) = \frac{\widehat{\rho_1}(\phi_1(\Tilde{X})) + \widehat{\rho_2}(\phi_2(\Tilde{X})) + \widehat{\rho_1}(\Tilde{X}-\phi_2(\Tilde{X})) + \widehat{\rho_2}(\Tilde{X}-\phi_1(\Tilde{X}))}{2}
\end{equation}
and to the following parameterizations of $f$ and $\ide - f$
$$
{f}_1(x):=  \frac{\phi_1(x) + x - \phi_2(x)}{2}, \quad  {f}_2(x):=  \frac{\phi_2(x) + x - \phi_1(x)}{2}=\ide(x)-f_1(x).
$$
Not only these are more robust estimates, but as they sum up to the identity, they provide acceptable allocations by construction.
Theorem \ref{thm:main_conv} guarantees the convergence of the induced optimal allocations for the risk-sharing problem. 

To train the neural networks and obtain the estimators of the theoretical optimum $(\Phi_1,\Phi_2):=(\hphi,\ide-\hphi)\in \lipone\times\lipone$, we minimize \eqref{eq:loss} with the optimizer Adam \citep{kingma2014adam}. The precise choices of the learning rate, batch size, the number of training epochs, and other implementation details are reported in the Appendix.
To ensure a robust framework, we train the DNNs multiple times and use their average as the final estimate. To be more explicit, we train ${f}_1$ and ${f}_2$ for $n$ times each starting from a different initialization. We obtain $n$ couples of neural networks $( {f}^k_1 , {f}^k_2)$ and use their arithmetic averages
\[\hat{\Phi}_1(\cdot) := \frac1n\sum_{k=1}^n{f}^k_1(\cdot)\,,\quad \hat{\Phi}_2(\cdot) := \frac1n\sum_{k=1}^n{f}^k_2(\cdot) \]
 to estimate $\Phi_1$ and $\Phi_2$, respectively. In all our experiments we chose $n=3$.

 To have a flexible framework, we allow the networks to have three types of different activation functions:
 \[\sigma(x)=\text{Tanh}(x)\,,\quad \sigma(x)=\text{ReLu}(x)\,,\quad\, \sigma(x) =x.\]
While the non-linear activation functions Tanh and ReLu are standard choices in deep learning, the reason for including the linear one will be apparent below. We incidentally note that, since $X$ is bounded, we have no issues in allowing for unbounded activation functions. We report a review of possible methodological and architectural enhancements in Appendix \ref{sub:PossEnh}.

 To verify the stability and reliability of our framework, we test our results with three different distributions
\begin{enumerate}
\item $X\sim\mathcal{U}[-1,1]$, where $\mathcal{U}$ is the uniform distribution;
\item $X\sim\mathcal{N}(0,1)$, where $\mathcal{N}$ is the normal distribution;
\item $X\sim-Beta(2,5)$, where $Beta$ is the $Beta$ distribution.
\end{enumerate}
The uniform distribution is the most basic example and it provides an easy setup to test our framework. A more interesting example is the normal distribution because of its well-known financial relevance. We note that in our experiments we restricted to $[-3,3]$ in order to have a distribution with bounded support. Finally, the $Beta$ distribution presents skewness and rare events, modeling the scarcity of data for extreme losses. We chose the opposite of a $Beta$ distribution in order to represent the financially more relevant case of pure losses.

\subsection{Initial tests: entropic and expected shortfall case}

To begin with, we test our framework in the well-known cases of entropic risk measures and expected shortfalls, for which explicit formulas are known. We start by recalling that, as in Examples 2.8 and 2.9 in \citep{FilipovicSvindland08},
\begin{equation}\label{eq:inf_conv_ES_entr}
	\entr_{\alpha}\square \entr_{\beta}(X)=\entr_{\alpha+\beta}(X),\qquad \ES_{\alpha}\square \ES_{\beta}(X)=\ES_{\alpha\vee\beta}(X).
\end{equation} This means that we can directly compute the theoretical value of the inf-convolution and compare it with the value $L_{\rho_1,\rho_2}(\Tilde{X})$ obtained by the DNNs. Additionally, we can calculate the $L^2$ error (under $\P_X$) between the estimated $\hat{\Phi}_1$ and $\hat{\Phi}_2$ and the theoretical ones.
As we show below, we found that the values of the inf-convolutions achieved by all our trained networks converge to the theoretical values and that  $\hat{\Phi}_i$ approximates ${\Phi}_i$ up to a negligible error, for $i=1,2$.
\begin{itemize} 
    \item For the entropic case, we chose $\rho_1(X)=\entr_{2}(X)$ and $\rho_2(X)=\entr_{3}(X)$ which yield the optimal allocation $\Phi_1(x)=\frac{2}{2 + 3}x$ and $\Phi_2(x)=\frac{3}{2 + 3}x$;
    \item For the $\ES$ case, we chose $\rho_1(X)=\ES_{0.8}(X)$ and $\rho_2(X)=\ES_{0.7}(X)$. It is clear from \eqref{eq:inf_conv_ES_entr} that $\Phi_1(x) =x$ and $\Phi_2(x) =0$ is an optimal allocation.
\end{itemize}

We start by discussing the entropic case. Figure \ref{fig:Phi_Ent_Norm} shows the comparison between the theoretical optimal allocations and the average predicted $\hat{\Phi}_1$ and $\hat{\Phi}_2$ for the normal distribution case and for the three activation functions. Every trained DNN seems to match perfectly the theoretical allocations. Indeed, we point out that the average predicted allocations in Figure \ref{fig:Phi_Ent_Norm} are plotted with their respective $\pm 3$ standard deviation bands across the $n$ networks. In particular, for this case, we notice that the uncertainty bands are invisible as they are almost null.

\begin{figure}
	\begin{center}
		\begin{minipage}{160mm}
			\subfigure{
				\resizebox*{8cm}{!}{\includegraphics{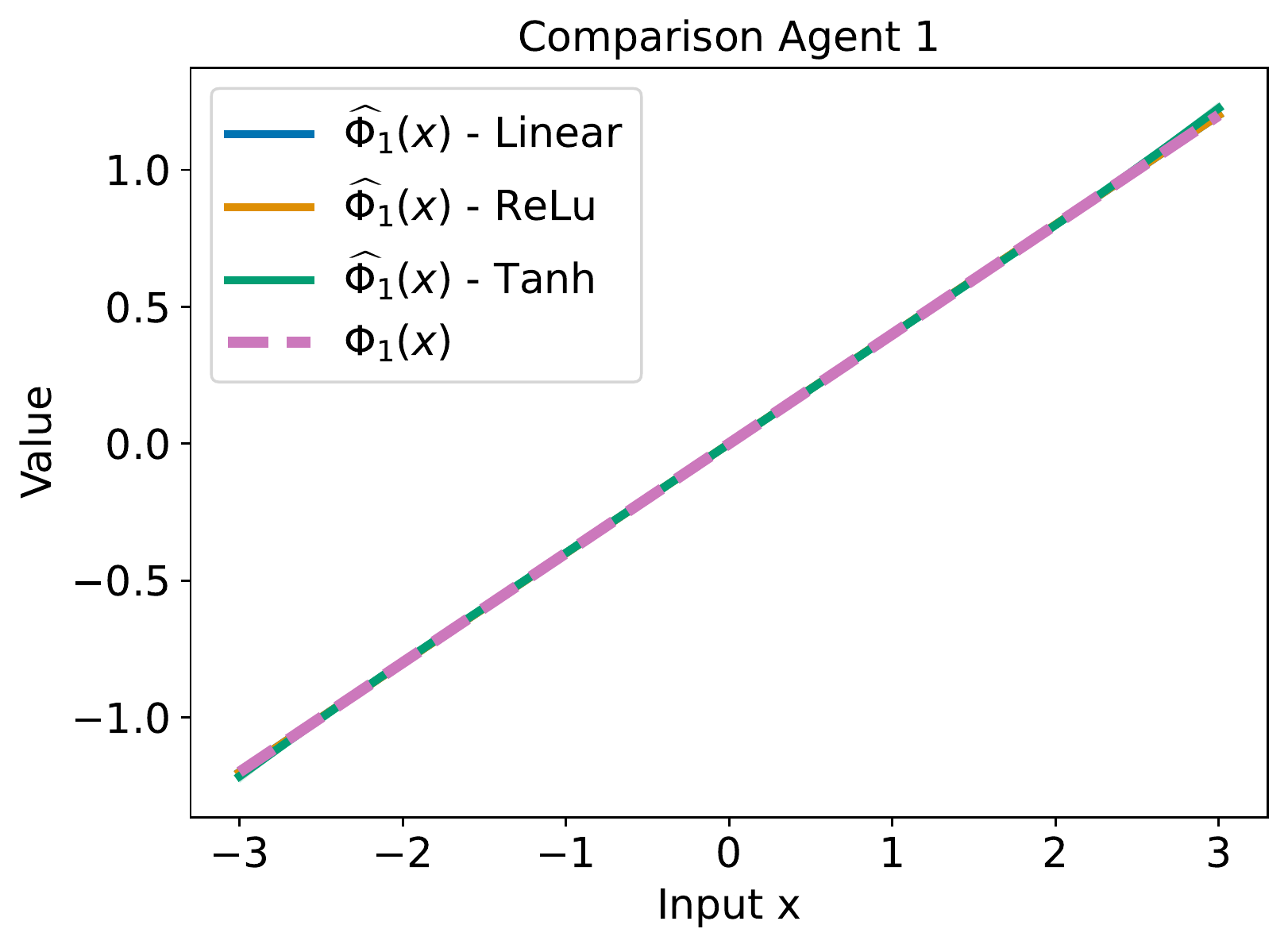}}}
			\subfigure{
				\resizebox*{8cm}{!}{\includegraphics{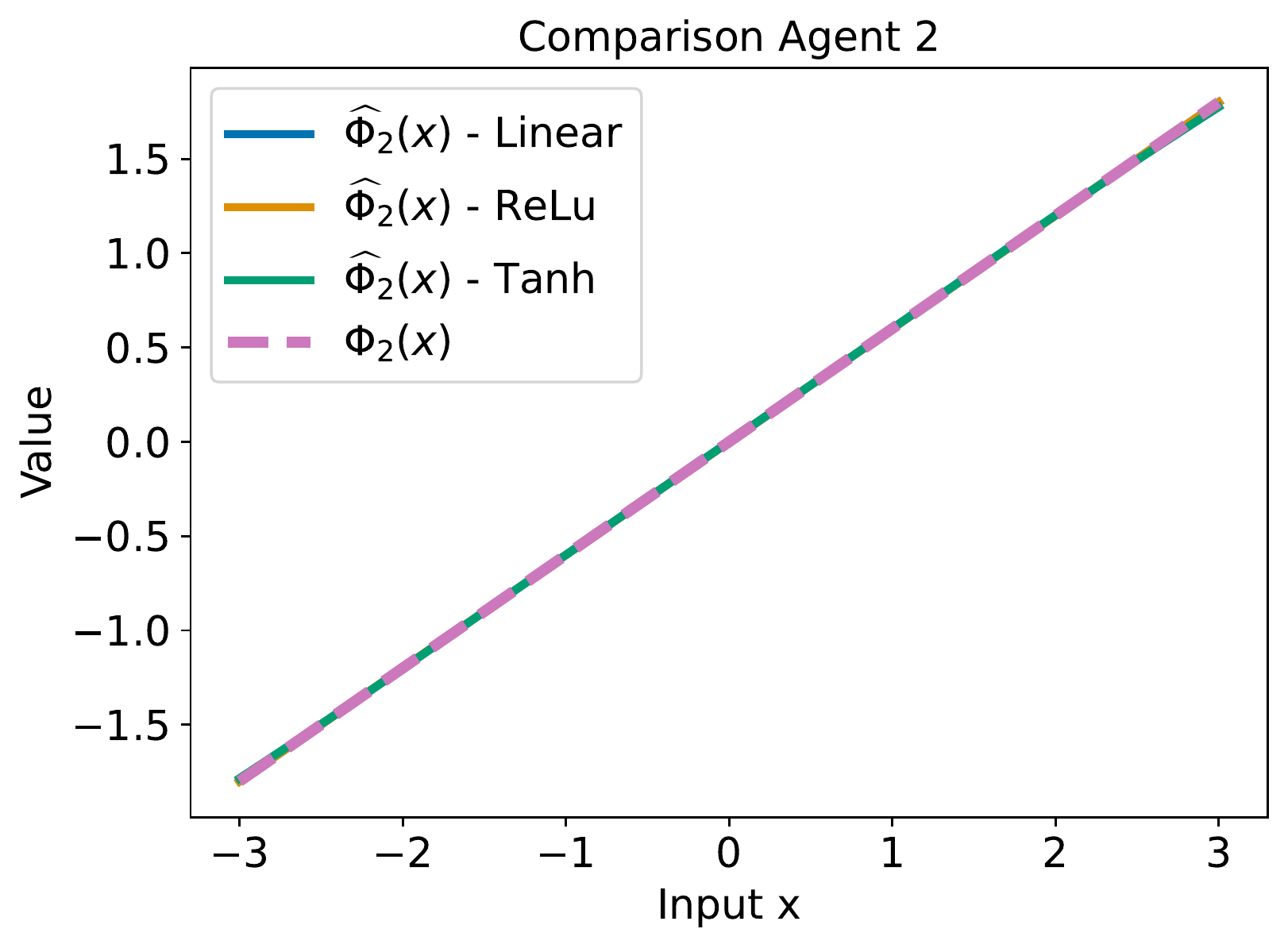}}}
			\caption{The entropic case - Normal distribution - Comparison of predicted vs. theoretical allocations.  We train all models over 3 trials and plot the average Predicted Allocation along with the $\pm 3$ standard deviations.
				\label{fig:Phi_Ent_Norm}}
		\end{minipage}
	\end{center}
\end{figure}
In Figure \ref{fig:Loss_Std_Ent_Norm_a}, we show the comparison between the average loss functions, along with the respective $\pm 3$ standard deviation shaded band, and the theoretical infimum calculated using \eqref{eq:inf_conv_ES_entr}. All three types of NNs achieve a loss that is close to the theoretical value of the inf-convolution, up to a negligible error. In Figure \ref{fig:Loss_Std_Ent_Norm_b}, for each model, we plot the standard deviation of the loss function \eqref{eq:loss}. Since the variance of the three losses is decreasing, we are observing a stable convergence.
\begin{figure}
	\begin{center}
		\begin{minipage}{160mm}
			\subfigure[Average training loss along with $\pm 3$ standard deviation.]{
				\resizebox*{8cm}{!}{\includegraphics{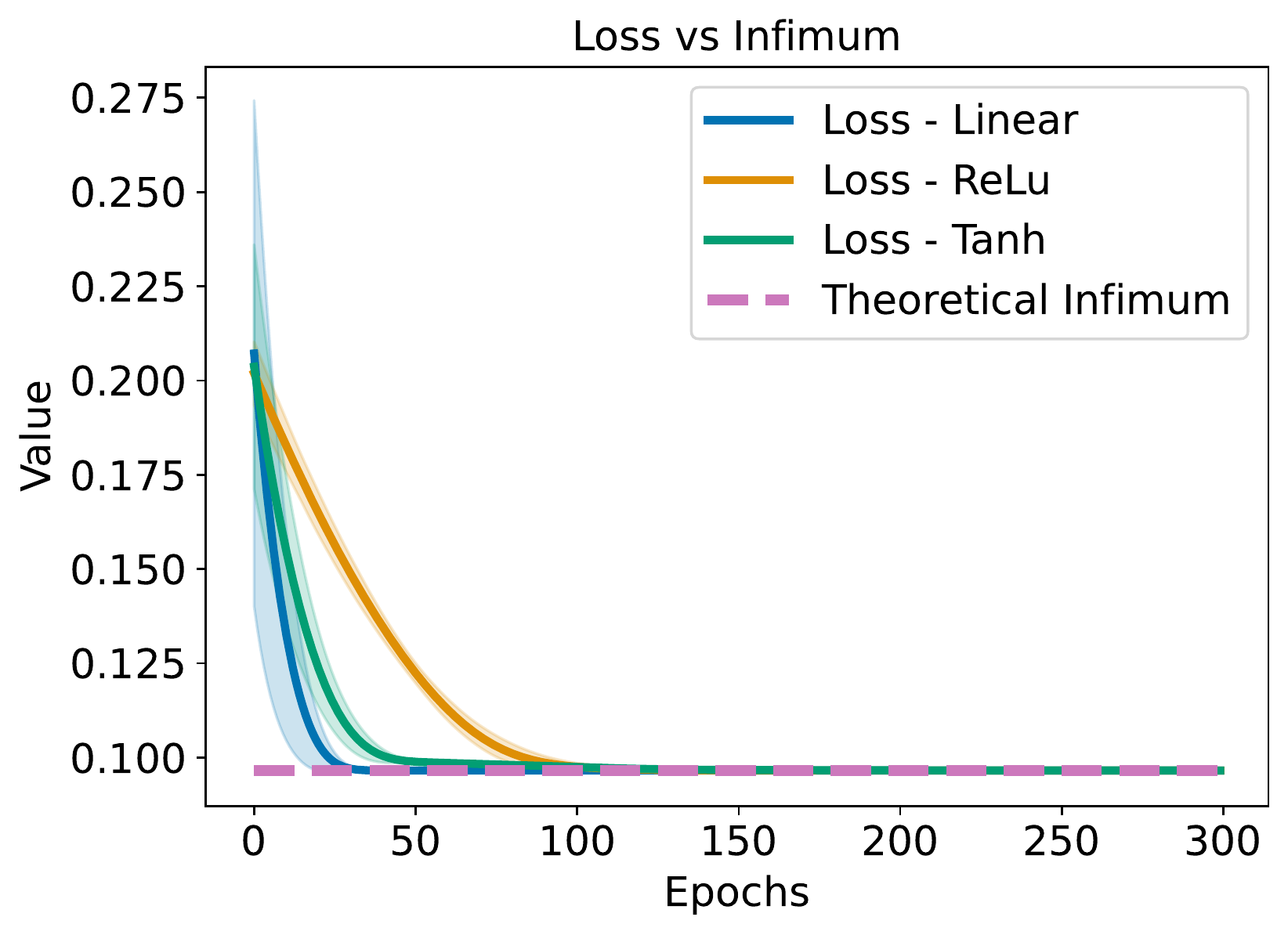}}\label{fig:Loss_Std_Ent_Norm_a}}
			\subfigure[Standard deviation of the loss.]{
				\resizebox*{8cm}{!}{\includegraphics{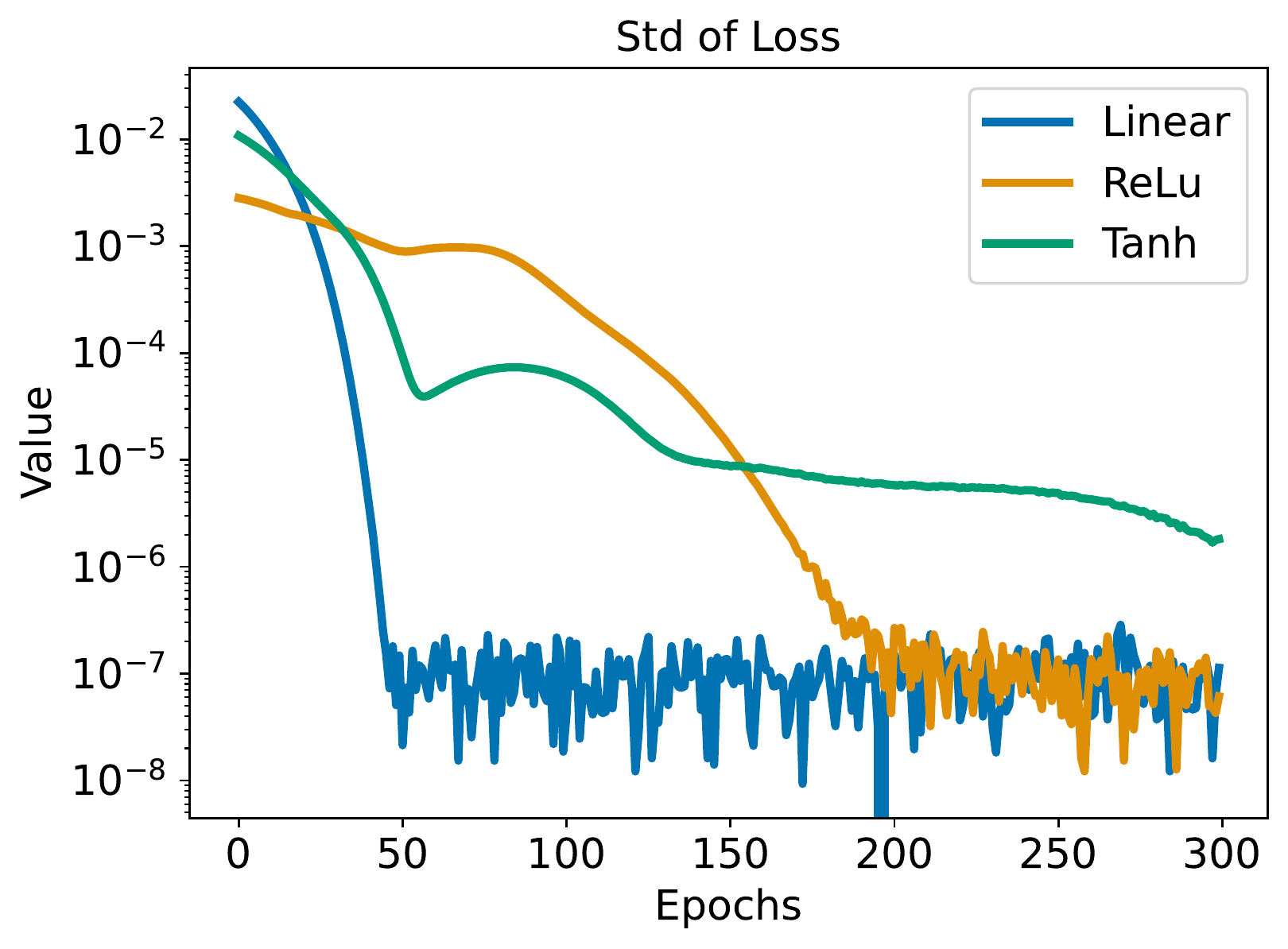}}\label{fig:Loss_Std_Ent_Norm_b}}
			\caption{The entropic case - Normal distribution - Convergence Analysis. 
				\label{fig:Loss_Std_Ent_Norm}}
		\end{minipage}
	\end{center}
\end{figure}

Table \ref{table:Entr_Normal} collects the data regarding the errors for the experiment with the standard normal distribution. We computed the average relative error with respect to the theoretical infimum, together with its standard deviation, and the $L^2$ error of $\hat{\Phi}_1$ with respect to the theoretical $\Phi_1$. Table \ref{table:Entr_full} reports the same figures also for the cases of uniform and $Beta$ distributions. We observe that the errors are all close to zero, meaning that all our NNs reached convergence and they exhibited low uncertainty, which is an indication of stable learning. Observe that for the entropic case (as well as for $\ES$) the optimal allocation is a linear function, therefore, it belongs to the span of the linear-activated DNN. We thus expect the linear activation to achieve the best performance. As we can appreciate in Table \ref{table:Entr_Normal}, this result is confirmed by our experiments. Additionally, we notice that also ReLu and Tanh are providing satisfactory performances.

\begin{table}
	\begin{center}
		\begin{tabular}{ |c|c|c|c| } 
\multicolumn{4}{c}{Entropic case - $\mathcal{N}(0,1)$  - Infimum = $0.09656$ }  \\
   \hline
 & Avg. Rel. Error & Std. Rel. Error & Avg. $L^{2}$ Error $\widehat{\Phi}_1$ \\
 \hline
Linear & $\% 1.800 \cdot 10^{-5} $ & $\% 1.172 \cdot 10^{-5}$ & $\mathbf{1.788 \cdot 10^{-7}}$\\
ReLu & $\% \mathbf{5.143 \cdot 10^{-6}}$ & $\% 6.341 \cdot 10^{-5}$ & $4.096 \cdot 10^{-5}$ \\
Tanh & $\% 5.955\cdot 10^{-6}$ & $\% \mathbf{1.814 \cdot 10^{-6}}$ & $5.813 \cdot 10^{-5}$ \\
\hline
\end{tabular}
\caption{Average relative errors between the losses and the theoretical infimum, their standard deviation, and the Average $L^{2}$ Error between $\widehat{\Phi}_1$ and $\Phi_1$.}
\label{table:Entr_Normal}
\end{center}
\end{table}

\begin{figure}
	\begin{center}
		\begin{minipage}{160mm}
			\subfigure{
				\resizebox*{8cm}{!}{\includegraphics{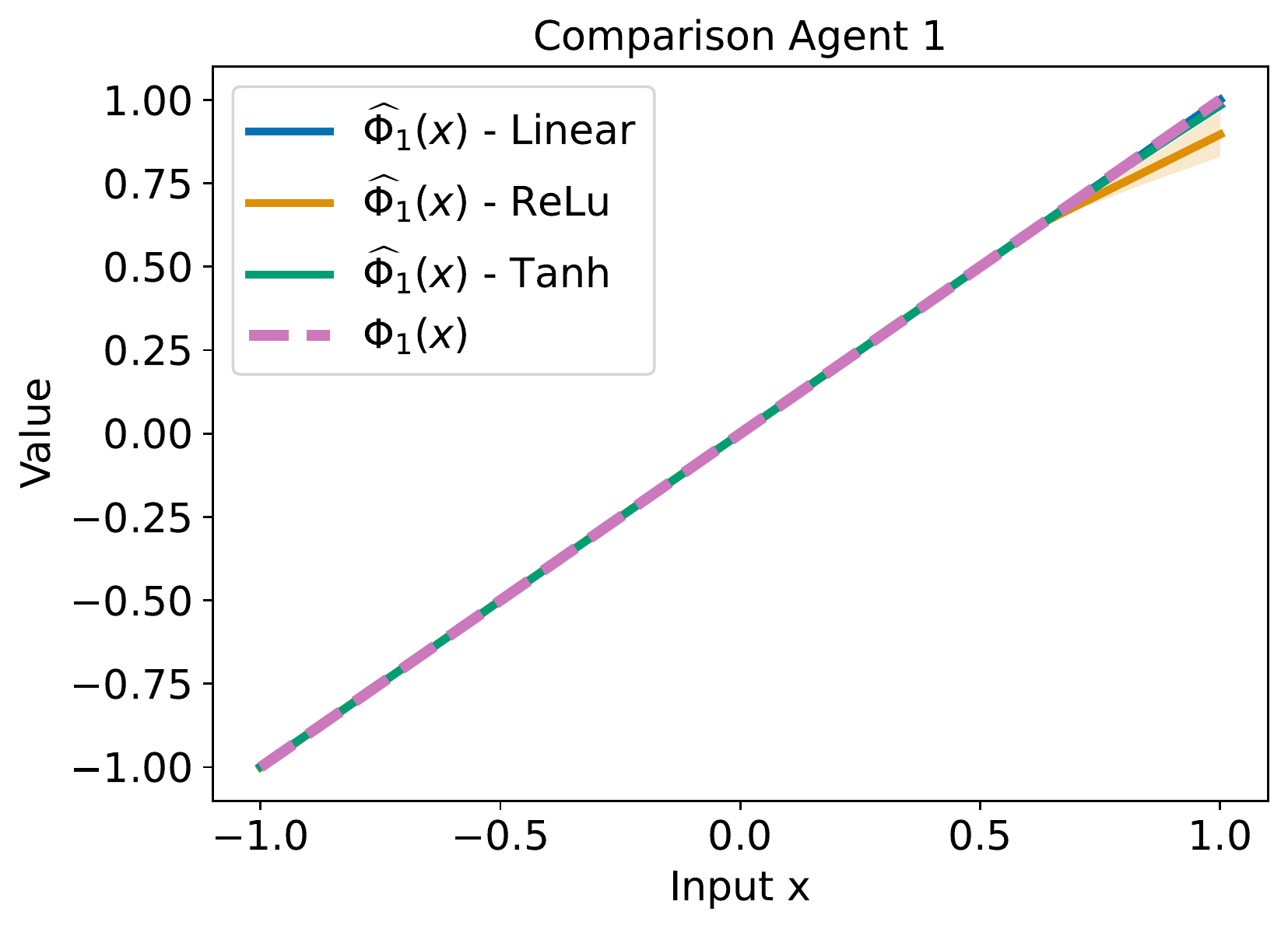}}}
			\subfigure{
				\resizebox*{8cm}{!}{\includegraphics{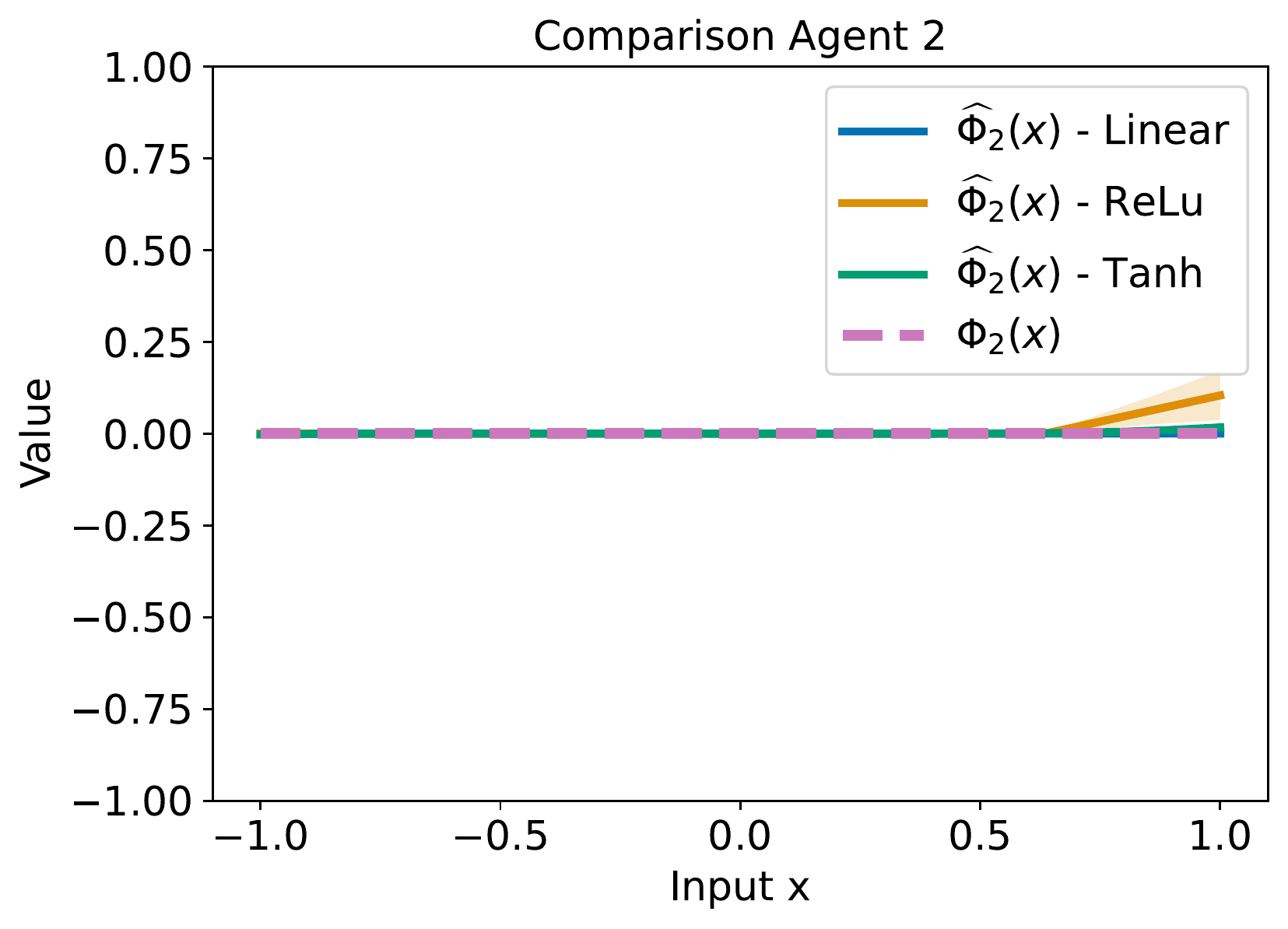}}}
			\caption{Expected Shortfall - Uniform distribution - Comparison of predicted vs. theoretical allocations.  We train all models over 3 trials and plot the average Predicted Allocation along with the $\pm 3$ standard deviations.
				\label{fig:Phi_ES_Unif}}
		\end{minipage}
	\end{center}
\end{figure}

\begin{figure}
	\begin{center}
		\begin{minipage}{160mm}
			\subfigure[Average training loss along with $\pm 3$ standard deviation.]{
				\resizebox*{8cm}{!}{\includegraphics{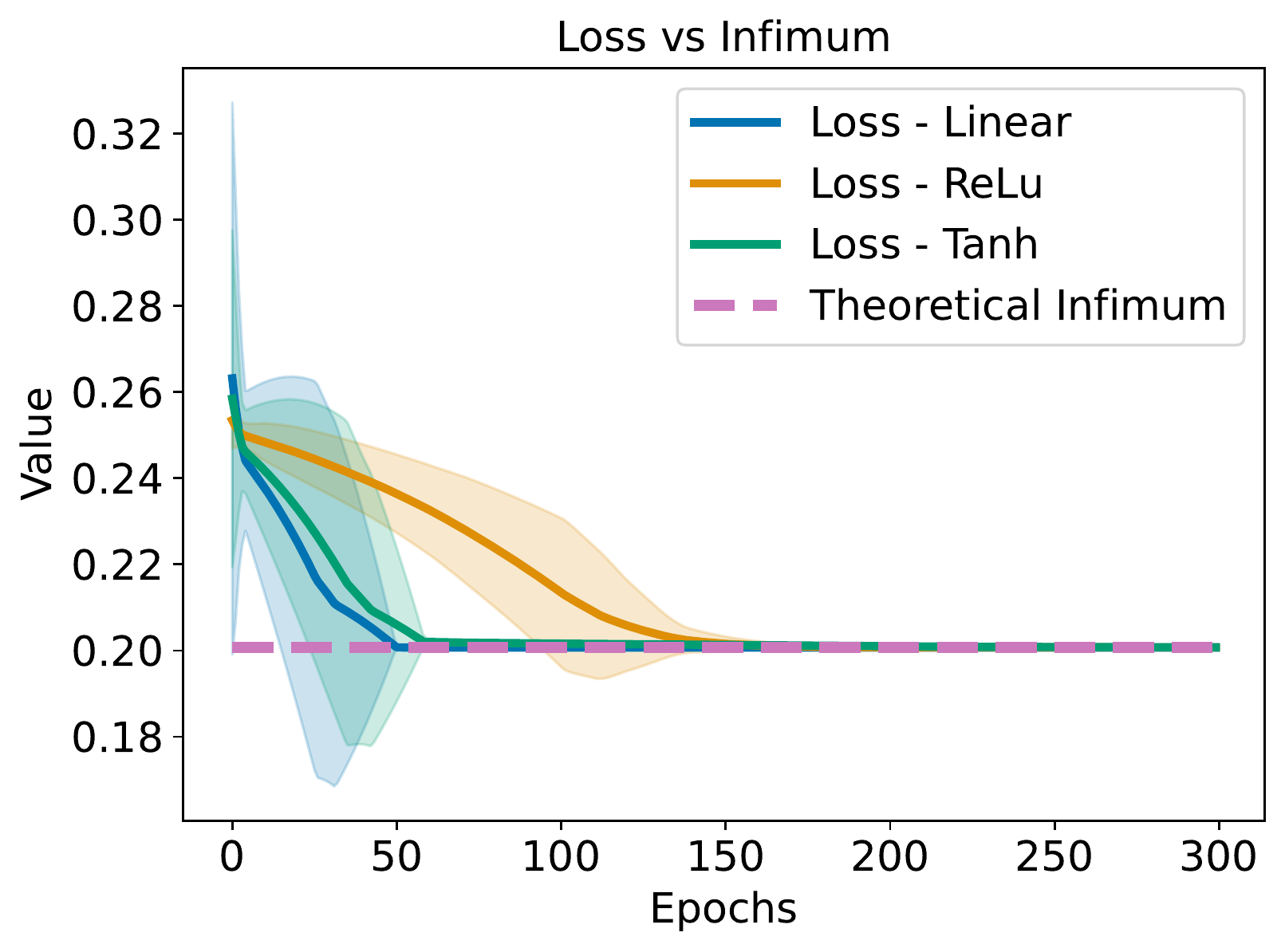}}\label{fig:Loss_Std_ES_Unif_a}}
			\subfigure[Standard deviation of the loss.]{
				\resizebox*{8cm}{!}{\includegraphics{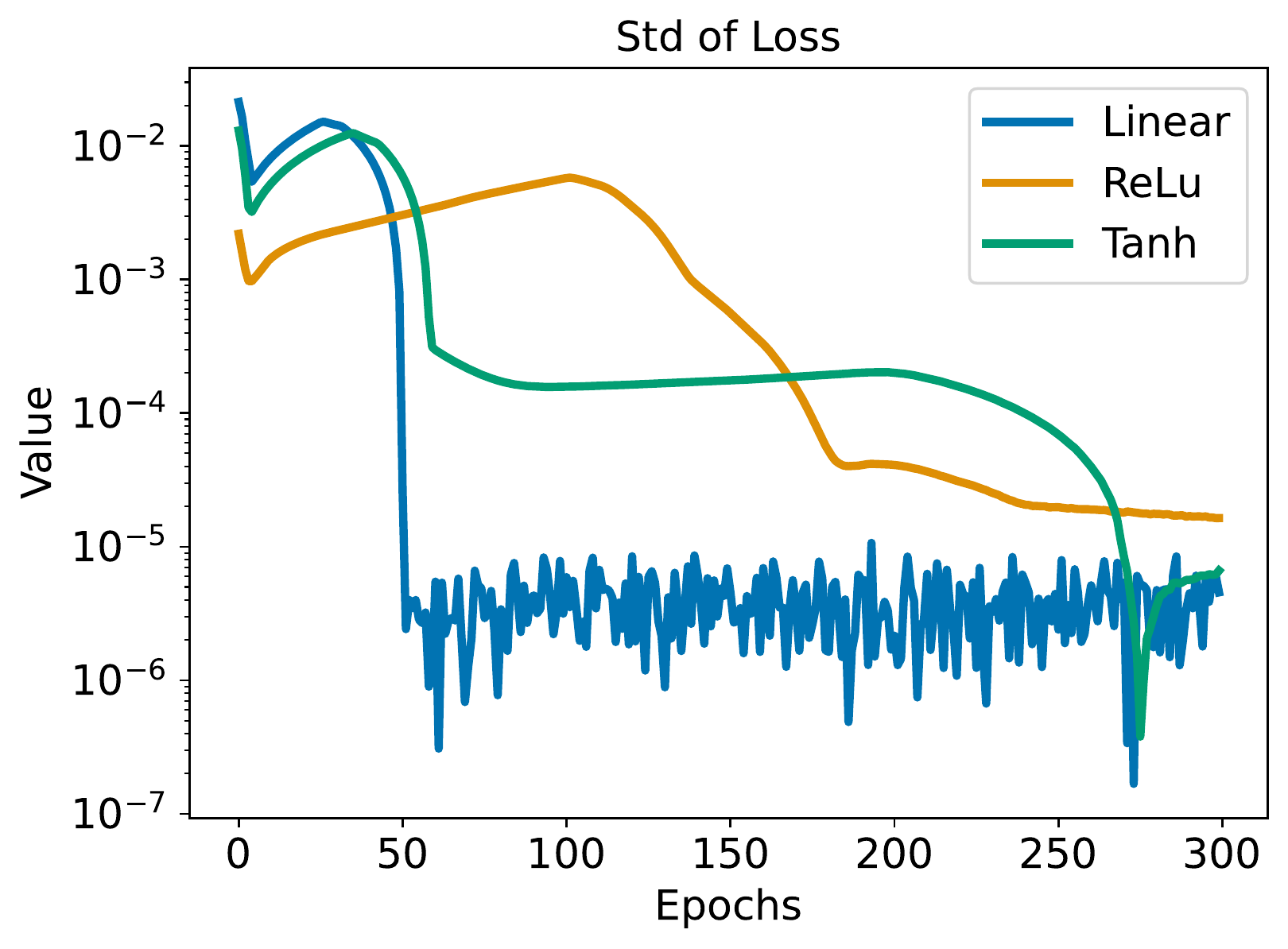}}\label{fig:Loss_Std_ES_Unif_b}}
			\caption{Expected Shortfall - Uniform distribution -  Convergence Analysis. 
				\label{fig:Loss_Std_ES_Unif}}
		\end{minipage}
	\end{center}
\end{figure}

Similar considerations apply to the case of $\ES$ and we obtain qualitatively and quantitatively the same results. As an example, in Figure \ref{fig:Phi_ES_Unif} we show the results for the uniform distribution case. Regarding the convergence analysis, we observe in Figure \ref{fig:Loss_Std_ES_Unif} that all three types of NNs achieve a loss that is only marginally distant from the theoretical value of the inf-convolution. Additionally, we notice that the convergence to such a value takes place with decreasing variance of the losses, indicating a stable convergence.
In Figure \ref{fig:Phi_ES_Unif} we present the comparison between the average predicted $\hat{\Phi}_1$ and $\hat{\Phi}_2$, along with their respective $\pm 3$ standard deviation shaded band, and the theoretical optimal allocations for the uniform distribution case. A consideration, which is specific to the $\ES$ case, is now due. The optimal allocation is of the form $\Phi_1(x) =x$ and $\Phi_2(x) =0$, and the DNN needs to learn the constant function in the latter case. Consistently with the known fact that using nonlinear functions for the (unsupervised) learning of constant functions is a challenging task, we find that ReLu and Tanh underperform with respect to the linear DNN.

In Table \ref{table:ES_Uniform}, we finally present the average relative error with respect to the theoretical infimum, its standard deviation, and $L^2$ error of $\hat{\Phi}_1$ with respect to $\Phi_1$. Table \ref{table:ES_full} reports the same figures also for the case of normal and $Beta$ distributions.

\begin{table}
	\begin{center}
		\begin{tabular}{ |c|c|c|c| } 
\multicolumn{4}{c}{Expected Shortfall  - $\mathcal{U}[-1,1])$  - Infimum = $0.2006$ }  \\
   \hline
 & Avg. Rel. Error & Std. Rel. Error & Avg. $L^{2}$ Error $\widehat{\Phi}_1$ \\
 \hline
Linear & $\mathbf{\% 2.846 \cdot 10^{-4}} $ & $\mathbf{\% 2.270 \cdot 10^{-3}}$ & $\mathbf{9.701 \cdot 10^{-6}}$\\
ReLu & $\% 1.379 \cdot 10^{-2} $ & $\% 8.171 \cdot 10^{-3}$ & $1.817 \cdot 10^{-2}$ \\
Tanh & $\% 2.742 \cdot 10^{-2}$ & $\% 3.301 \cdot 10^{-3}$ & $2.521 \cdot 10^{-3}$ \\
\hline
\end{tabular}
\caption{Average relative errors between the losses and the theoretical infimum, their standard deviation, and the Average $L^{2}$ Error between $\widehat{\Phi}_1$ and $\Phi_1$.}
\label{table:ES_Uniform}
\end{center}
\end{table}

\subsection{Convolution of distortion risk measures}
\label{secsemiexpl}
We here consider the case where both $\rho_1$ and $\rho_2$ are distortion risk measures, as in \eqref{def:distRM}, with respect to some discrete probabilities $\mu_1$, $\mu_2$. Let $N_1$, $N_2$ be two given integers, and consider the risk measures \[\rho_1(X) :=\sum_{j=1}^{N_1}\mu_{1j} \ES_{\alpha_{1j}} (X),\qquad \rho_2(X) :=\sum_{j=1}^{N_2}\mu_{2j} \ES_{\alpha_{2j}} (X),\] where $\mu_{ij}>0$ with $\sum_{j=1}^{N_i}\mu_{ij}=1$ and  $0<\alpha_{ij}<1$ for $j=1,\ldots,N_i$ and for $i=1,2$.
Some semi-explicit expressions of the optimal allocations are known for this case, in particular, an optimal allocation can be found as a linear combination of  ReLu functions, possibly composed with translation maps --- see Example 3.1 in \citep{JST07} and also Appendix A of \citep{ELW18} for a more general case\footnote{We thank an anonymous referee for pointing out this fact.}. Hence, we expect the ReLu-activated DNN to achieve the best performance. 

Differently from the entropic and $\ES$ cases, the problem has a non-linear solution and we expect the linear-activated DNN to perform poorly. Nevertheless, for the sake of consistency in our tests, we included the linear activation in all experiments. As an example, we chose
$$\rho_1(X)= 0.5\ES_{0.8}(X) + 0.5\ES_{0.7}(X), \quad \rho_2(X)= 0.7\ES_{0.9}(X) + 0.3\ES_{0.5}(X).$$
Figure \ref{fig:Phi_ESES_Beta} shows the average predicted $\hat{\Phi}_1$ and $\hat{\Phi}_2$ for the case of the $Beta$ distribution and for the three activation functions. As we can observe, the DNNs trained with non-linear activation functions agree on the shape of the solution, whereas the linear-activated one is clearly different. 

\begin{figure}
	\begin{center}
		\begin{minipage}{160mm}
			\subfigure{
				\resizebox*{8cm}{!}{\includegraphics{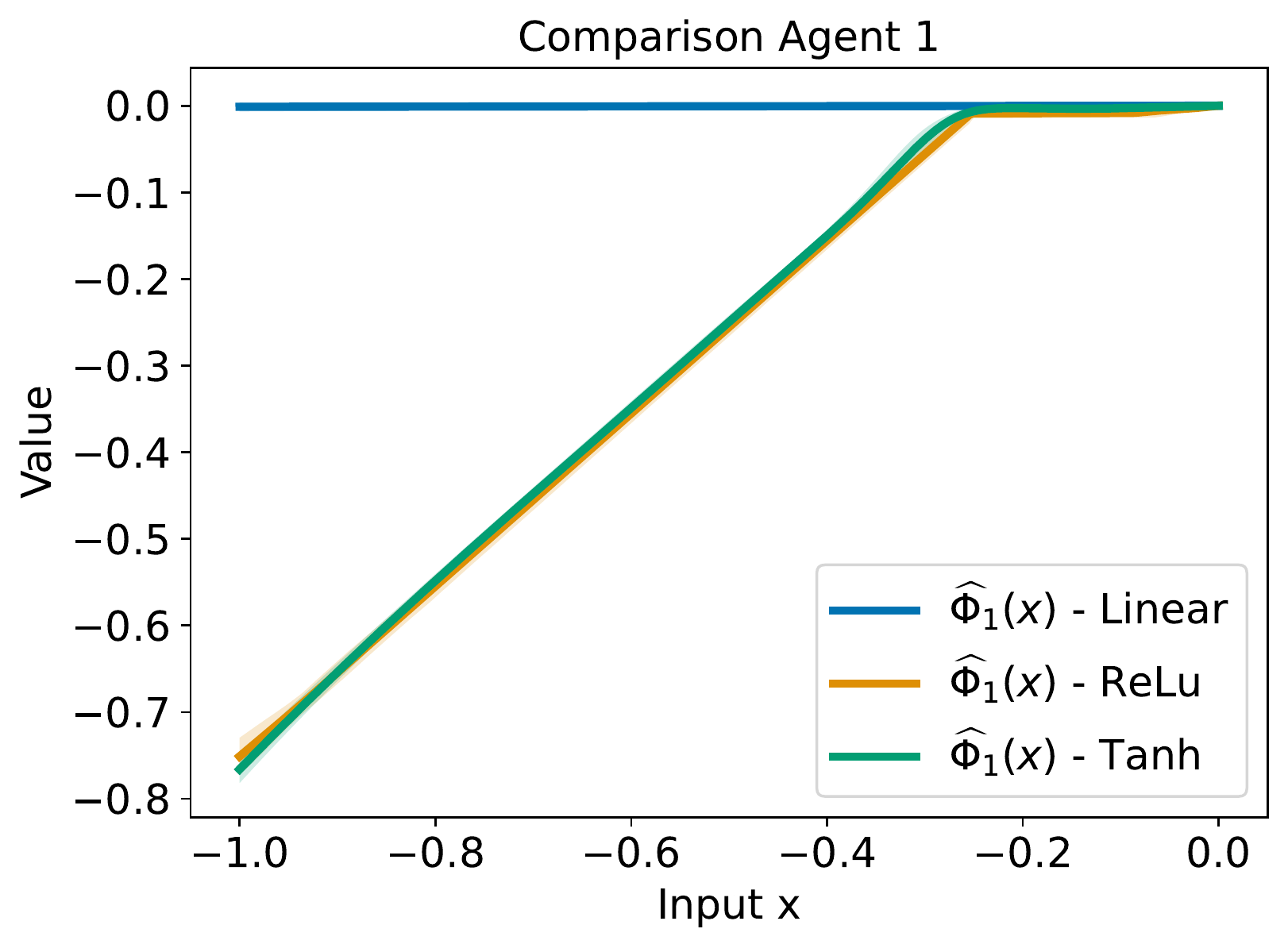}}}
			\subfigure{
				\resizebox*{8cm}{!}{\includegraphics{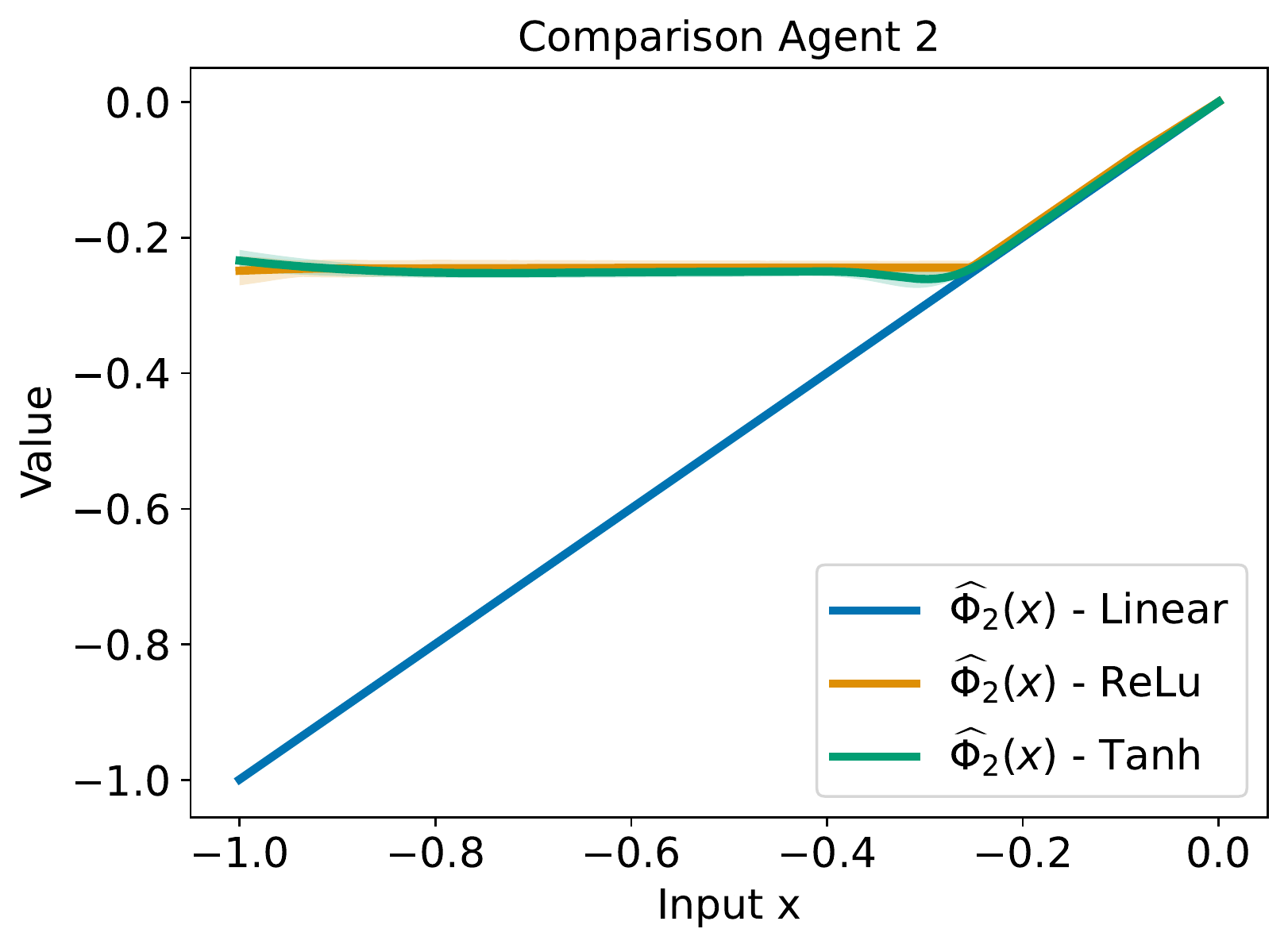}}}
			\caption{Distortion Measures - Beta distribution - Predicted allocations.  We train all models over 3 trials and plot the average Predicted Allocation along with the $\pm 3$ standard deviations.
				\label{fig:Phi_ESES_Beta}}
		\end{minipage}
	\end{center}
\end{figure}
\begin{figure}
	\begin{center}
		\begin{minipage}{160mm}
			\subfigure[Average training loss along with $\pm 3$ standard deviation.]{
				\resizebox*{8cm}{!}{\includegraphics{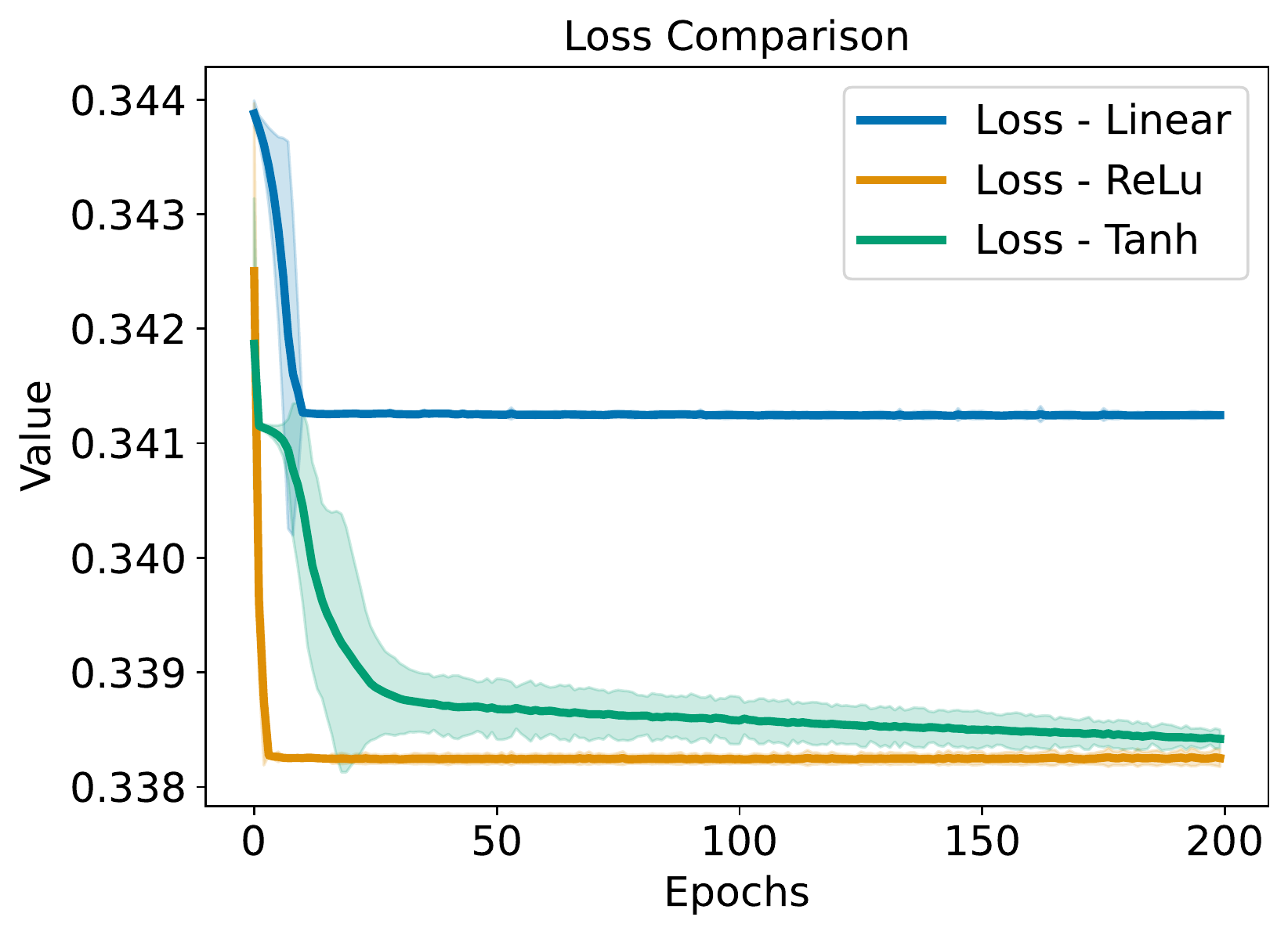}}\label{fig:Loss_Std_ESES_Beta_a}}
			\subfigure[$\hat{\Phi}_1(x)$ for ReLu-activated DNN as a weighted sum of ReLu activations.]{
				\resizebox*{8cm}{!}{\includegraphics{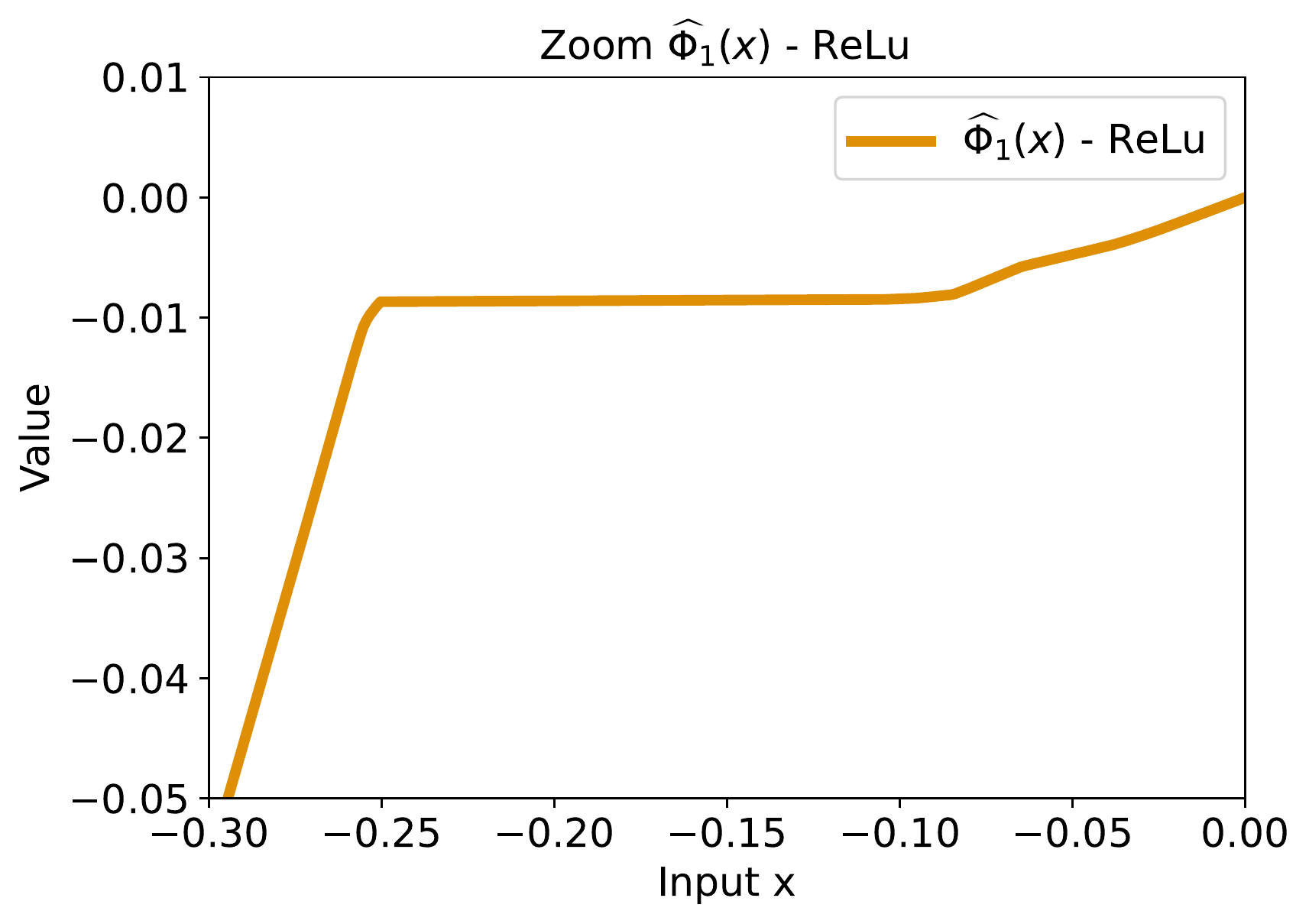}}\label{fig:Loss_Std_ESES_Beta_b}}
			\caption{Distortion Measures - Beta distribution - Convergence Analysis and optimal ReLu. 
				\label{fig:Loss_Std_ESES_Beta}}
		\end{minipage}
	\end{center}
\end{figure}

As anticipated before, we expect the optimal allocations to be linear combinations of ReLu activations. Consistently with the theory, if we look at the average predicted $\widehat{\Phi}_1(x)$ for the case of ReLu, as in Figure \ref{fig:Loss_Std_ESES_Beta_b}, we observe that such expected behavior is captured.
Figure \ref{fig:Loss_Std_ESES_Beta_a} shows the average loss functions as a function of the training epochs. First of all, we notice that the linear NN achieves a loss level that is sensibly larger than those achieved by the ReLu and the Tanh DNNs, confirming the expectations of it poor performance. Secondly, we notice that the loss decreases with decreasing variance, indicating a stable convergence with low uncertainty in all three cases.  

Finally, in Table \ref{table:ESES_full}, we present the average achieved losses, together with the uncertainty of their estimates, for all activation functions and for all distributions. 
In line with the theoretical predictions, the DNN activated with a ReLu function is the one performing best in terms of average loss: all three loss values are, by construction, greater or equal to the theoretical infimum, and the best performance is understood in the sense of achieving the lowest value. The Tanh-activated DNN is comparably reliable. From Table \ref{table:ESES_full}, we can see that in some cases the linear-activated DNN shows the most stable convergence, namely the lowest standard deviation of losses. However, it converges to a loss value that is significantly higher than the other two. This is not unexpected since, by design, the linear-activated DNN is unable to represent a nonlinear function and, therefore, exhibits poorer performances.

\begin{table}
	\begin{center}
		\begin{tabular}{ |c|c|c| } 
\multicolumn{3}{c}{Distortion Measure - $\mathcal{U}[-1,1])$  }  \\
   \hline
 & Avg. Loss & Std. Loss  \\
 \hline
Linear & $ 0.220786 $ & $ \mathbf{1.79364 \cdot 10^{-7}} $\\
\textbf{ReLu}&  $ \mathbf{0.210493} $ & $ 1.43078 \cdot 10^{-6}$  \\
Tanh &  $ 0.210722 $ & $ 4.81990 \cdot 10^{-5}$  \\
\hline
\multicolumn{3}{c}{}\\
\multicolumn{3}{c}{Distortion risk measures - $\mathcal{N}(0,1)$  }  \\
   \hline
 & Avg. Loss & Std. Loss  \\
 \hline
Linear & $ 0.371297 $ & $ \mathbf{1.12391 \cdot 10^{-7}} $\\
\textbf{ReLu} &  $ \mathbf{0.355218} $ & $ 2.53271 \cdot 10^{-7}$ \\
Tanh &  $ 0.355505 $ & $ 3.12838 \cdot 10^{-5}$   \\
\hline
\multicolumn{3}{c}{}\\
\multicolumn{3}{c}{Distortion Measure - $-Beta(2,5)$}  \\
   \hline
 & Avg. Loss & Std. Loss  \\
 \hline
Linear & $ 0.341245 $ & $ 9.27940 \cdot 10^{-6} $\\
\textbf{ReLu} &  $ \mathbf{0.338251} $ & $ \mathbf{2.52160 \cdot 10^{-6}}$ \\
Tanh &  $ 0.338419 $ & $ 2.66384 \cdot 10^{-6}$  \\
\hline
\end{tabular}
\caption{Average Loss of the achieved training losses together with its standard deviation.}
\label{table:ESES_full}
\end{center}
\end{table}

\subsection{Heterogeneous agents}
\label{sec:hetero}
In our last experiments, we consider two heterogeneous agents, in the sense that one adopts an entropic risk measure while the other one opts for a distortion-type risk measure. In the first of such experiments, the risk measures are 
\begin{equation}
\label{Esconentr1}
    \rho_1(X)= \ES_{0.9}(X), \quad \rho_2(X)= \entr_{0.3}(X).
\end{equation}

From \citep{JST07} Proposition 3.2 or \citep{Ruschendorf13} Theorem 11.22, the optimal allocation is known to be induced by $(f,\ide-f)=(-(x-k)^{-}, \max(x,k))$ for some (non-explicit) constant $k$.
In line with the previous subsections, we show an example of the average predicted $\hat{\Phi}_1$ and $\hat{\Phi}_2$. In Figure \ref{fig:Phi_ESEntr_Beta_Big}, we plot the predicted allocations for the $Beta$ distribution, for the three different activation functions. Once again, we expect the solution to be non-linear and we can observe that the optimal allocations found by ReLu and Tanh-activated DNNs are comparable, whereas the one found by the linear-activated DNN differs significantly. In Figure \ref{fig:Loss_Std_ESEntr_Beta_Big_b} we isolated the allocation $\hat{\Phi}_1$ found by the ReLu DNN which, as we will see below, is the one that performed best. We notice that the desired behavior of the optimal allocations is well-captured.

\begin{figure}
	\begin{center}
		\begin{minipage}{160mm}
			\subfigure{
				\resizebox*{8cm}{!}{\includegraphics{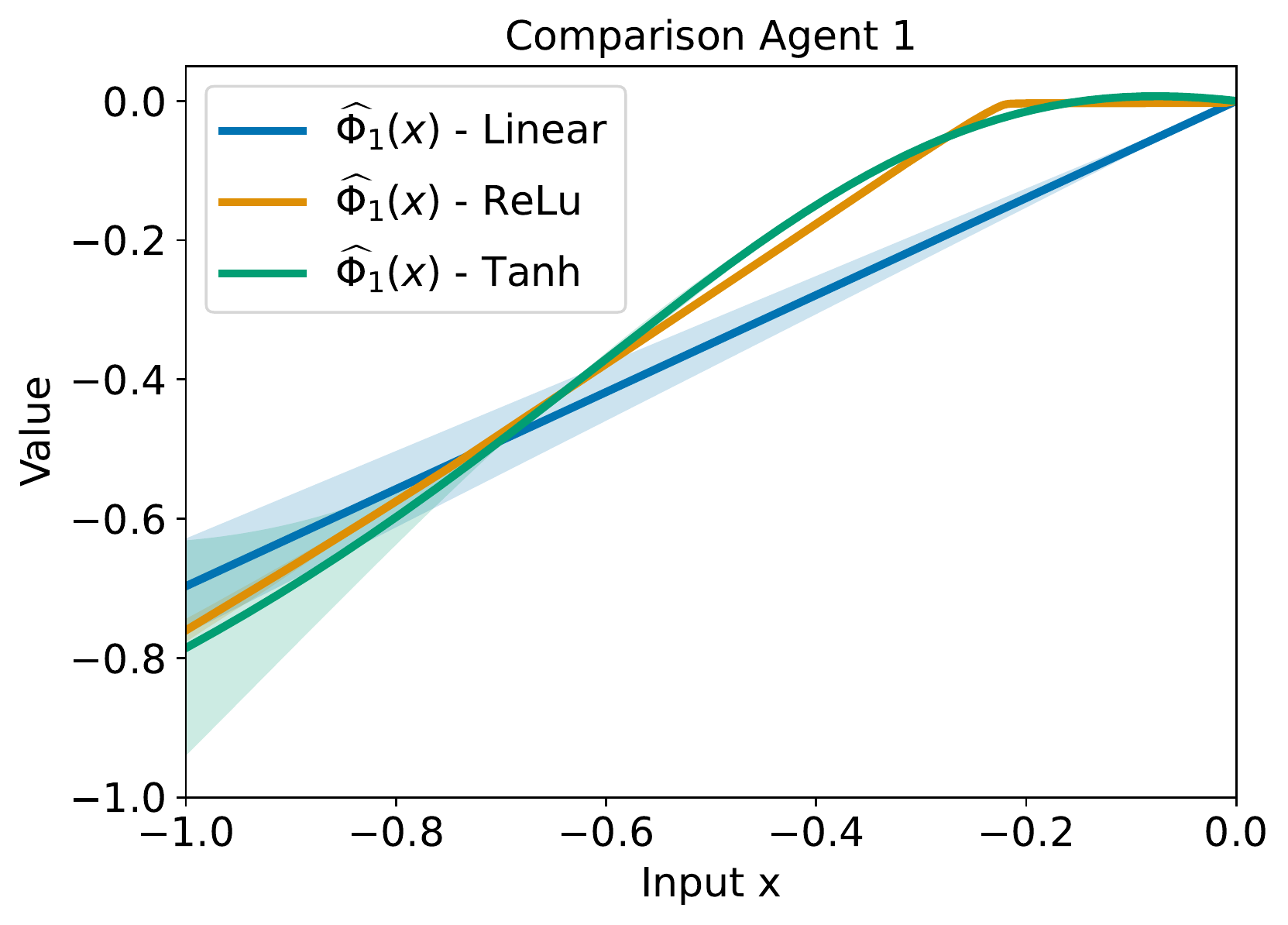}}}
			\subfigure{
				\resizebox*{8cm}{!}{\includegraphics{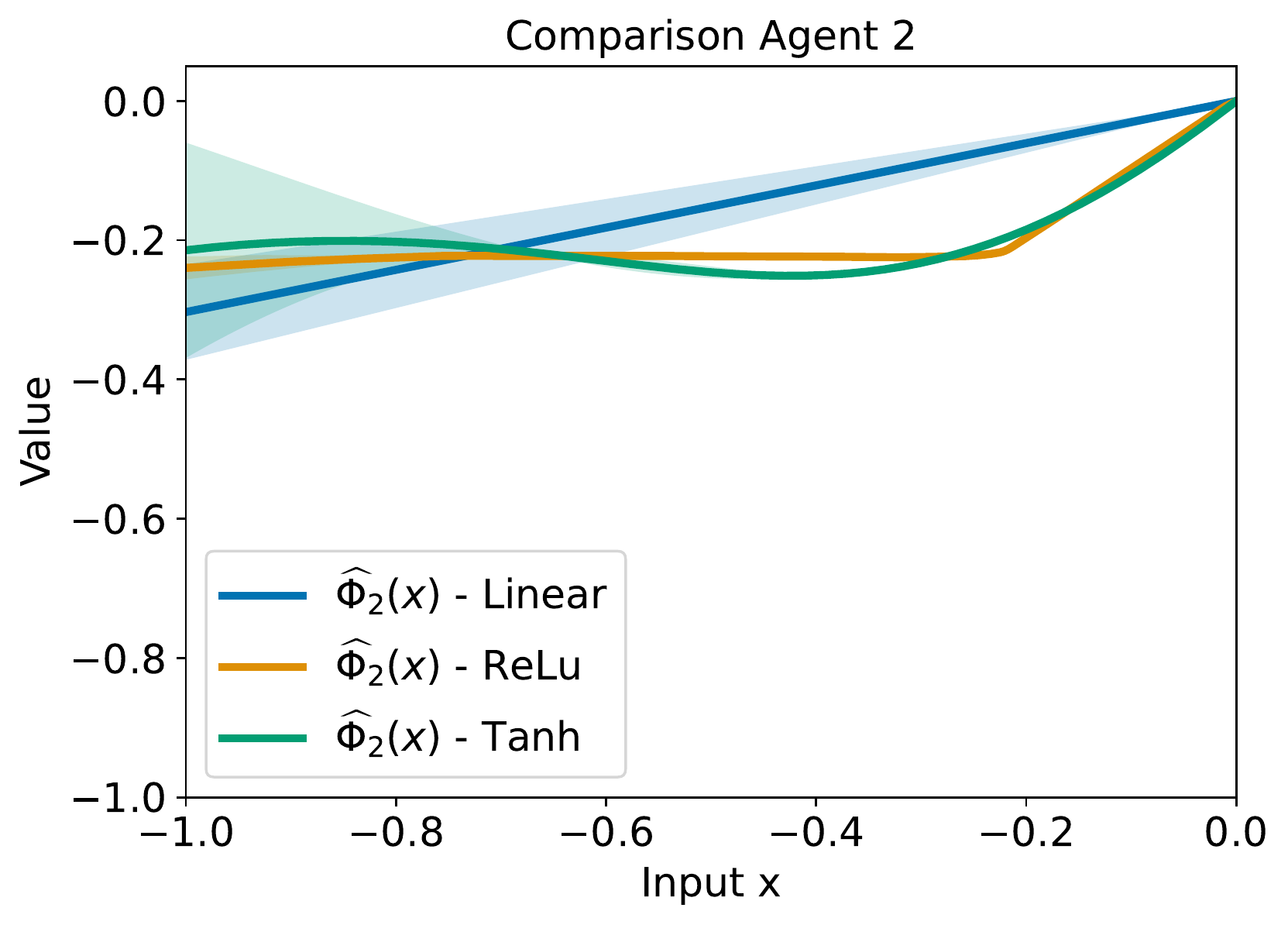}}}
			\caption{Heterogeneous Agents - Case Eq. \eqref{Esconentr1} - Beta distribution - Predicted allocations.  We train all models over 3 trials and plot the average Predicted Allocation along with the $\pm 3$ standard deviations.
				\label{fig:Phi_ESEntr_Beta_Big}}
		\end{minipage}
	\end{center}
\end{figure}

\begin{figure}
	\begin{center}
		\begin{minipage}{160mm}
			\subfigure[Average training loss along with $\pm 3$ standard deviation.]{
				\resizebox*{8cm}{!}{\includegraphics{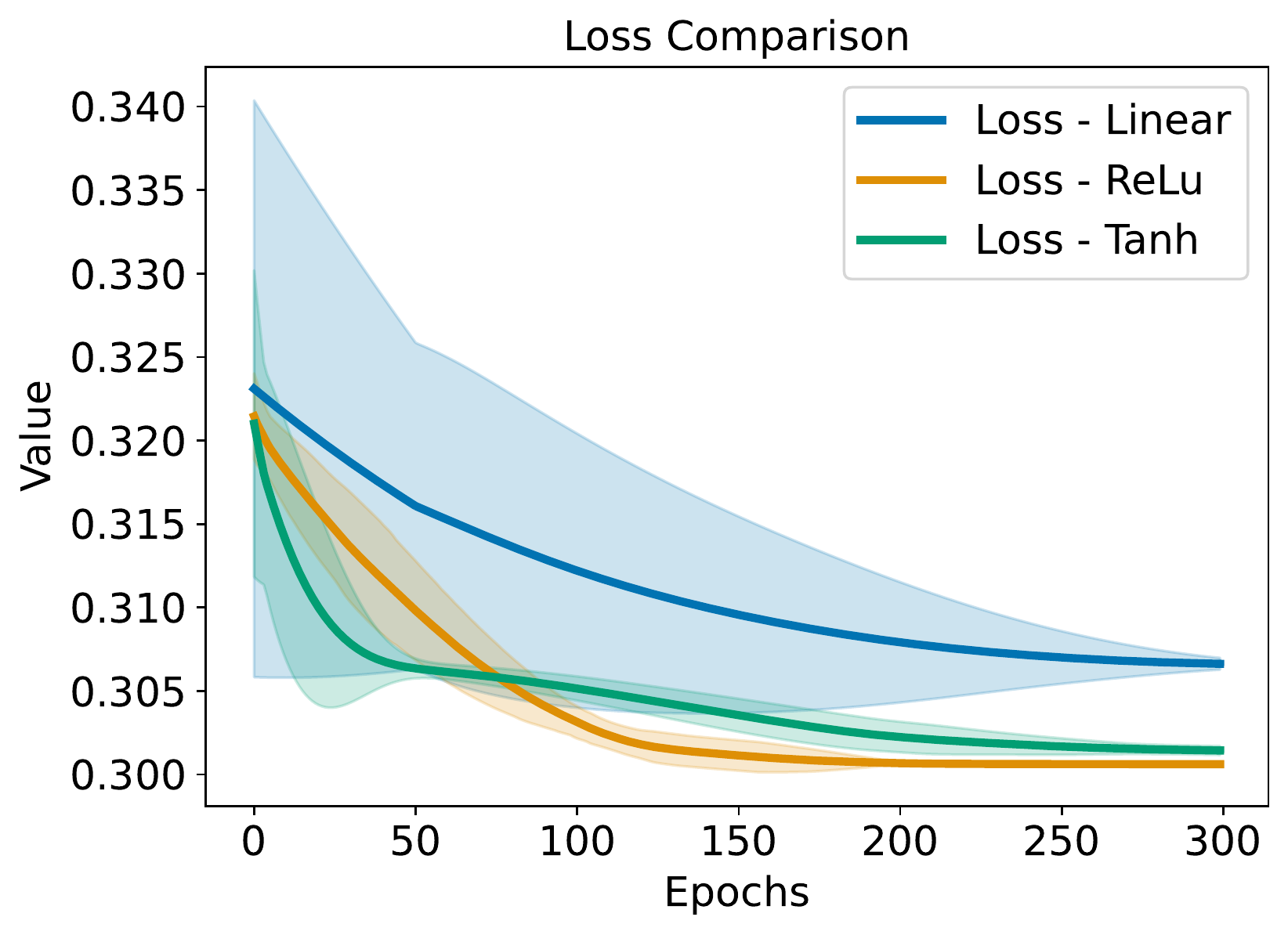}}\label{fig:Loss_Std_ESEntr_Beta_Big_a}}
			\subfigure[$\hat{\Phi}_1(x)$ for ReLu-activated DNN is (almost) as per theoretical prediction]{
				\resizebox*{8cm}{!}{\includegraphics{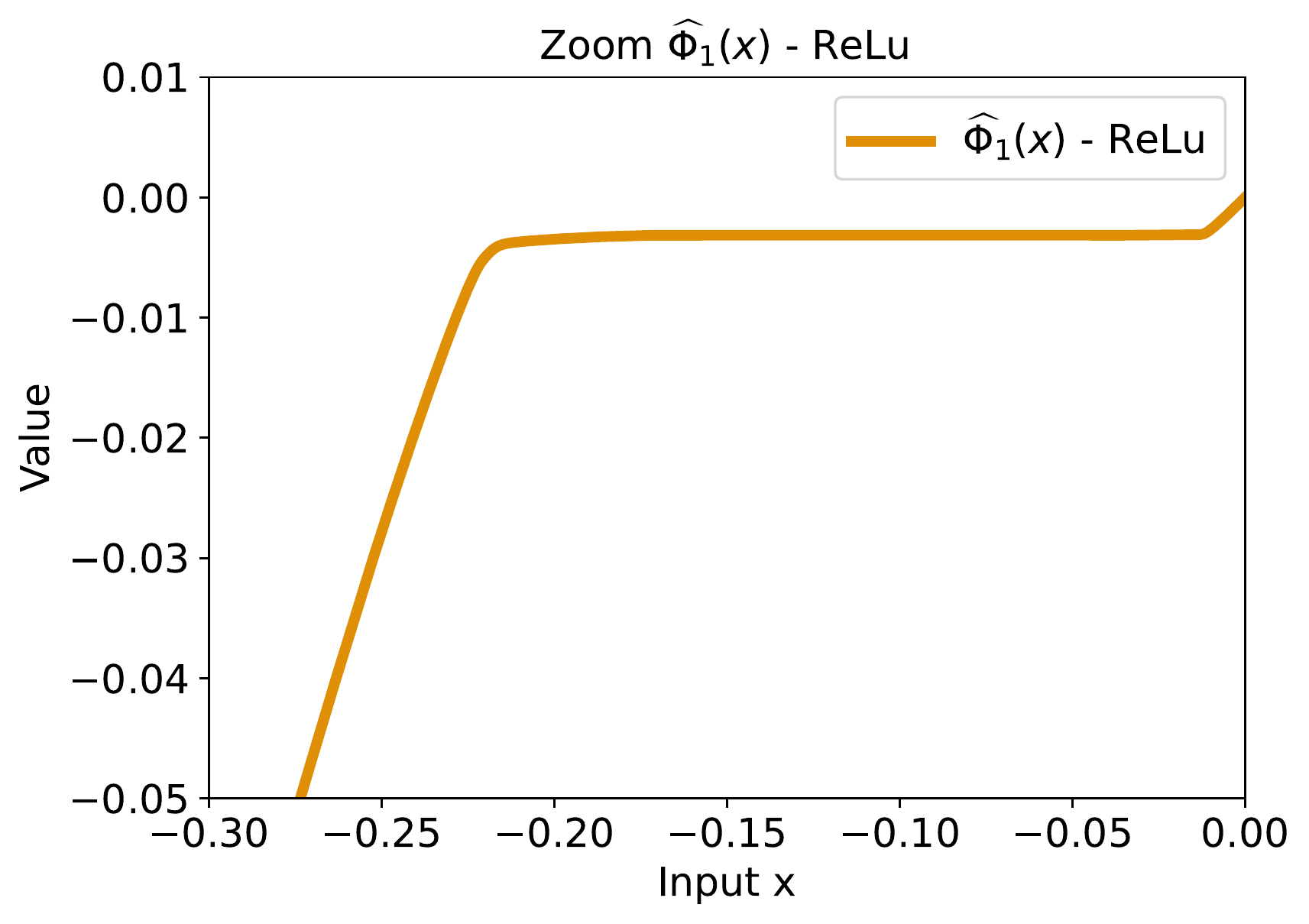}}\label{fig:Loss_Std_ESEntr_Beta_Big_b}}
			\caption{Heterogeneous Measures - Case Eq. \eqref{Esconentr1} - Beta distribution - Convergence Analysis and optimal ReLu. 
				\label{fig:Loss_Std_ESEntr_Beta_Big}}
		\end{minipage}
	\end{center}
\end{figure}

 Figure \ref{fig:Loss_Std_ESEntr_Beta_Big_a} shows the average loss functions as a function of the training epochs, together with their uncertainty-shaded bands. We notice that all networks exhibit stable convergence. However, the linear-activated network achieves a loss level that is sensibly larger than those achieved by the ReLu and Tanh-activated ones. From the picture, it is already clear that Relu is the one performing best in this case. This is confirmed by the data that we collect in Table \ref{table:ESEntr_full}, namely, the average achieved loss together with the uncertainty of their estimates. Nevertheless, we note that while the Tanh NN underperforms with respect to the ReLu one, it provides comparable performances.

\begin{table}
	\begin{center}
		\begin{tabular}{ |c|c|c| } 
\multicolumn{3}{c}{Case Eq. \eqref{Esconentr1} - $\mathcal{U}[-1,1])$  }  \\
   \hline
 & Avg. Loss & Std. Loss  \\
 \hline
Linear & $ 0.0962926 $ & $ \mathbf{1.62093 \cdot 10^{-6}} $\\
\textbf{ReLu}&  $ \mathbf{0.0837376} $ & $ 7.15505 \cdot 10^{-6}$  \\
Tanh &  $ 0.085397 $ & $ 7.64700 \cdot 10^{-5}$ \\
\hline
\multicolumn{3}{c}{}\\
\multicolumn{3}{c}{Case Eq. \eqref{Esconentr1} - $\mathcal{N}(0,1)$  }  \\
   \hline
 & Avg. Loss & Std. Loss  \\
 \hline
Linear &$ 0.185575 $ & $ \mathbf{1.65563 \cdot 10^{-6}} $ \\
\textbf{ReLu} &  $ \mathbf{0.166919} $ & $ 2.50091 \cdot 10^{-4}$ \\
Tanh &  $ 0.169095 $ & $ 9.12596 \cdot 10^{-6}$   \\
\hline
\multicolumn{3}{c}{}\\
\multicolumn{3}{c}{Case Eq. \eqref{Esconentr1} - $-Beta(2,5)$}  \\
   \hline
 & Avg. Loss & Std. Loss  \\
 \hline
Linear & $ 0.306628 $ & $ 1.16732 \cdot 10^{-4} $\\
\textbf{ReLu} &  $ \mathbf{0.300616} $ & $ \mathbf{2.28480 \cdot 10^{-6}}$ \\
Tanh &  $ 0.301437 $ & $ 8.87485 \cdot 10^{-5}$   \\
\hline
\end{tabular}
\caption{Average Loss of the achieved training losses together with its standard deviation.}
\label{table:ESEntr_full}
\end{center}
\end{table}

In our last experiment, we consider a case where, to the best of our knowledge, no theoretical information is available.
Again, we consider two heterogeneous agents, the first one opting for a distortion risk measure, and the second one adopting an entropic risk measure. More precisely, the risk measures are 
\begin{equation}
\label{Entrconmixedes}
  \rho_1(X)= 0.7\ES_{0.8}(X) + 0.3\ES_{0.7}(X), \quad \rho_2(X)= \entr_{0.3}(X).  
\end{equation}

\begin{figure}
	\begin{center}
		\begin{minipage}{160mm}
			\subfigure{
				\resizebox*{8cm}{!}{\includegraphics{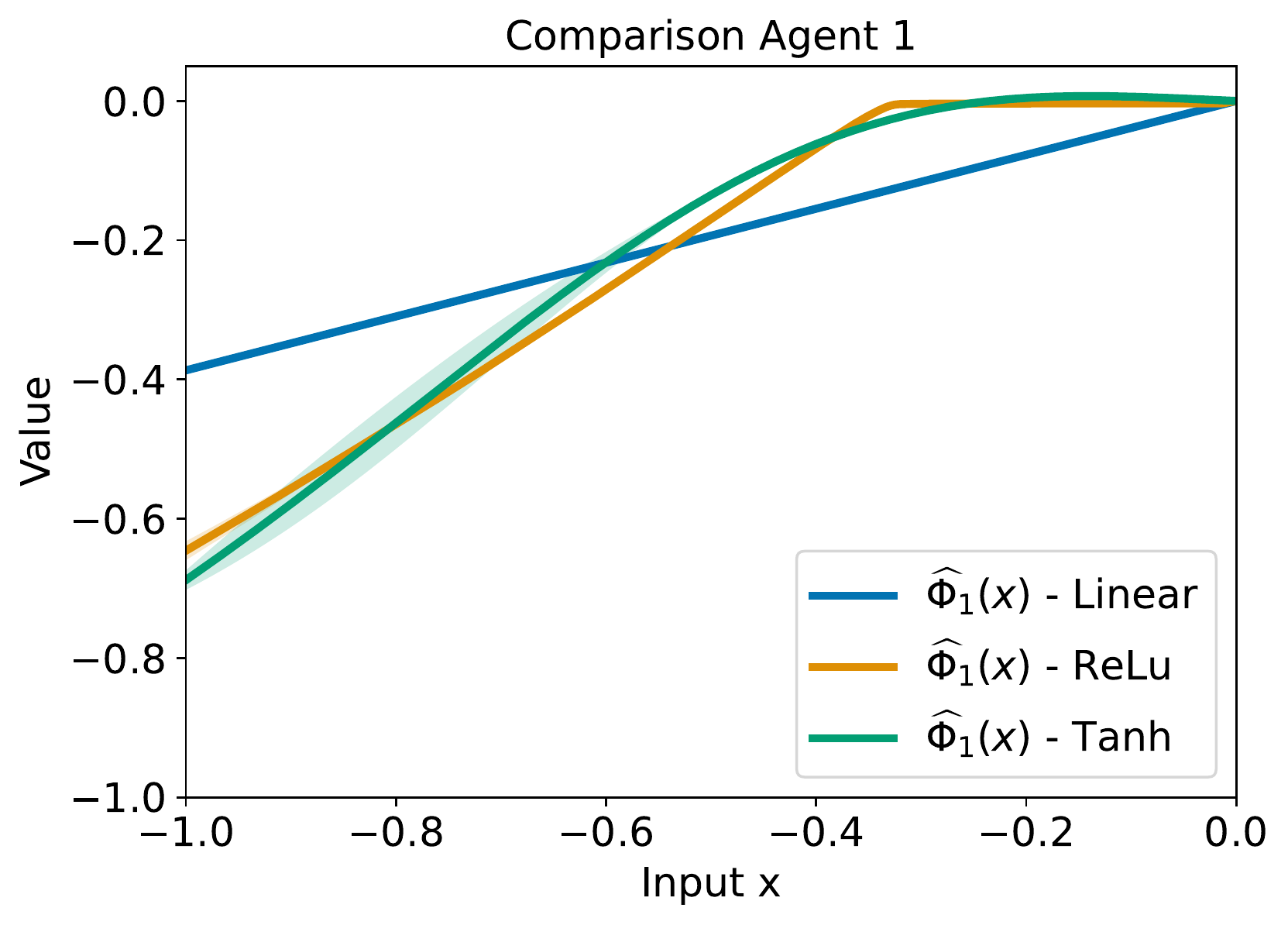}}}
			\subfigure{
				\resizebox*{8cm}{!}{\includegraphics{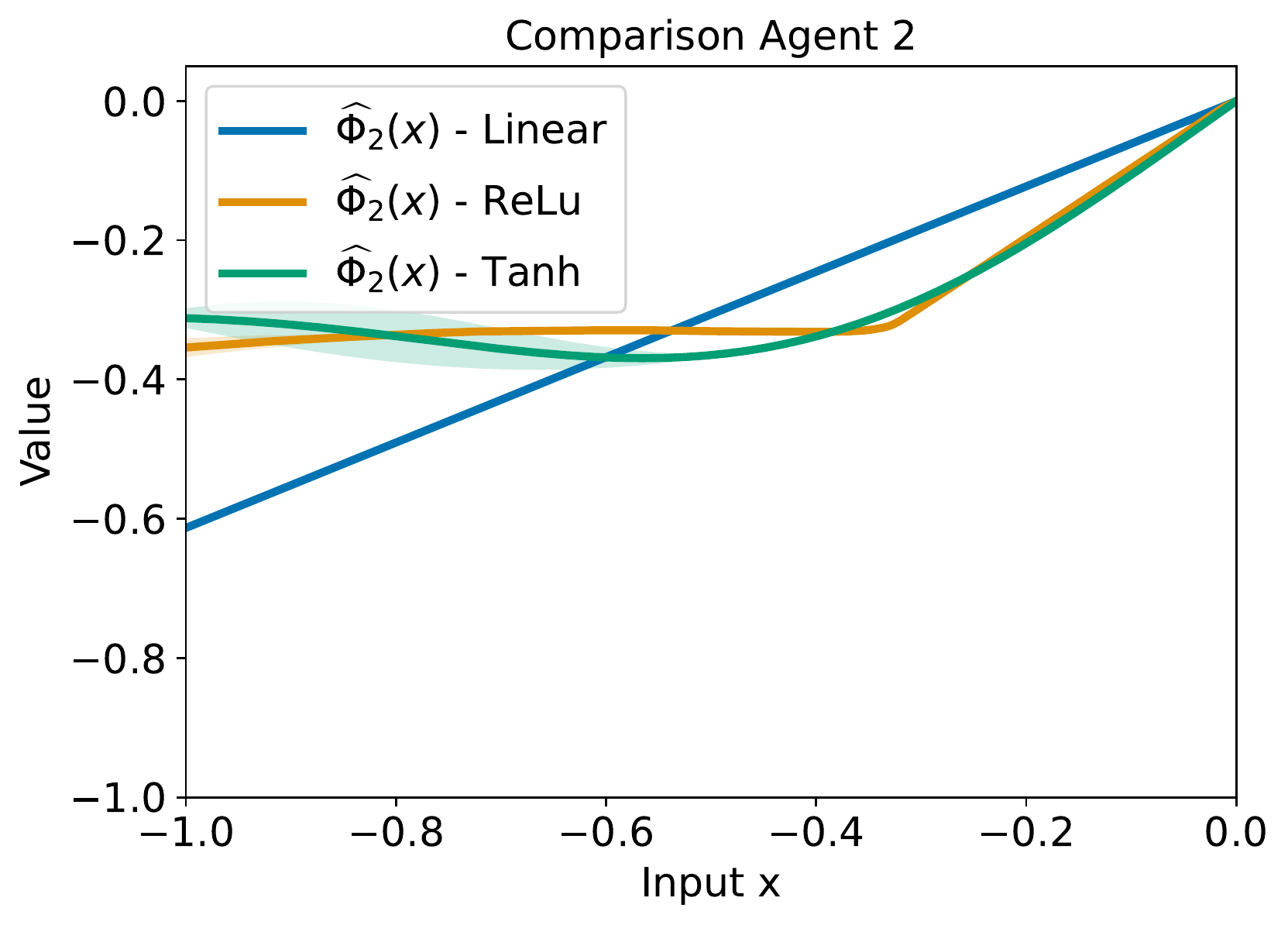}}}
			\caption{Heterogeneous Agents - Case Eq. \eqref{Entrconmixedes} - Beta distribution - Predicted allocations.  We train all models over 3 trials and plot the average Predicted Allocation along with the $\pm 3$ standard deviations.
				\label{fig:Phi_ESESEntr_Beta_Big}}
		\end{minipage}
	\end{center}
\end{figure}

\begin{figure}
	\begin{center}
		\begin{minipage}{160mm}
			\subfigure[Average training loss along with $\pm 3$ standard deviation.]{
				\resizebox*{8cm}{!}{\includegraphics{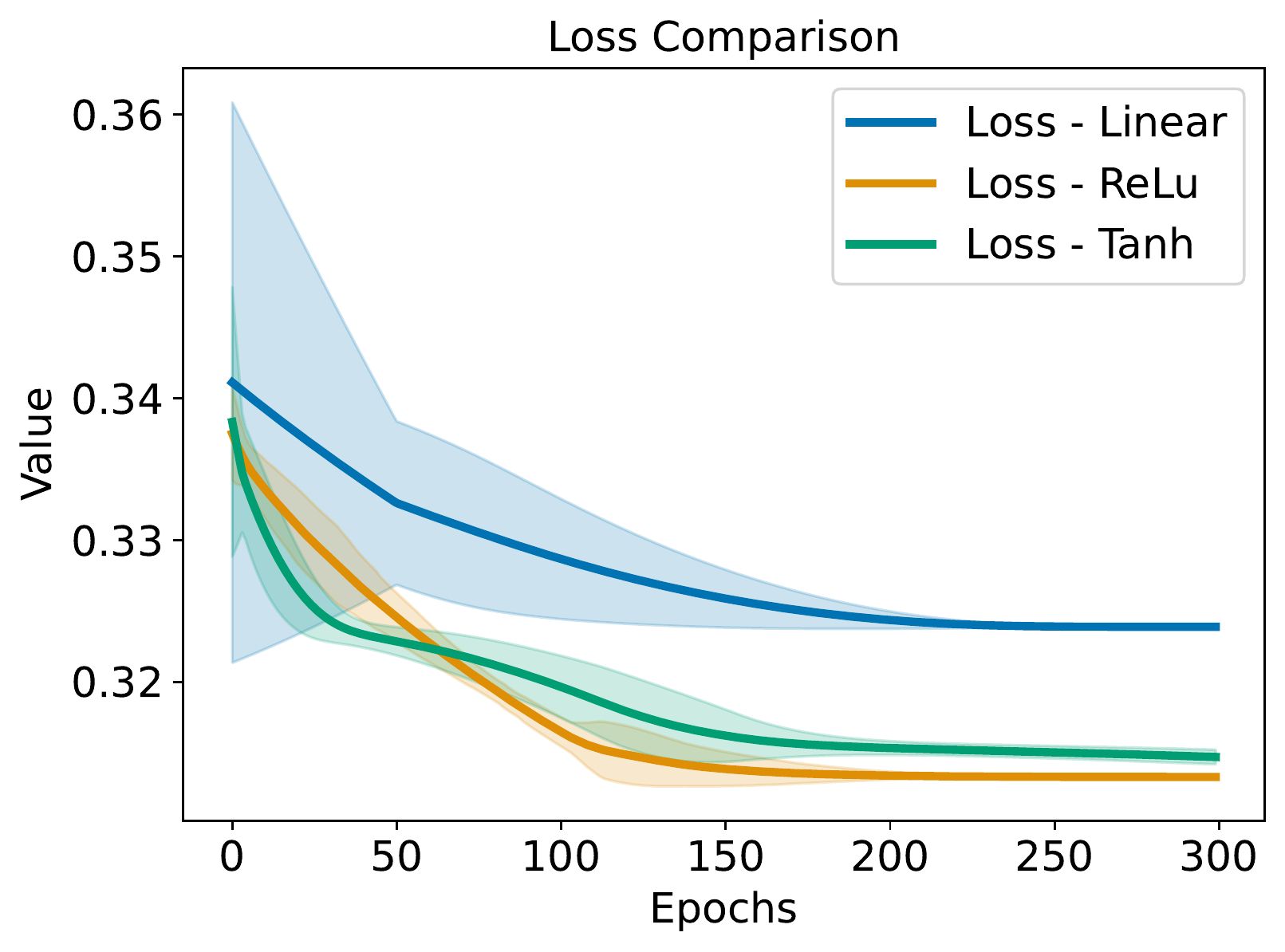}}\label{fig:Loss_Std_ESESEntr_Beta_Big_a}}
			\subfigure[$\hat{\Phi}_1(x)$ for ReLu-activated DNN as a weighted sum of ReLu functions.]{
				\resizebox*{8cm}{!}{\includegraphics{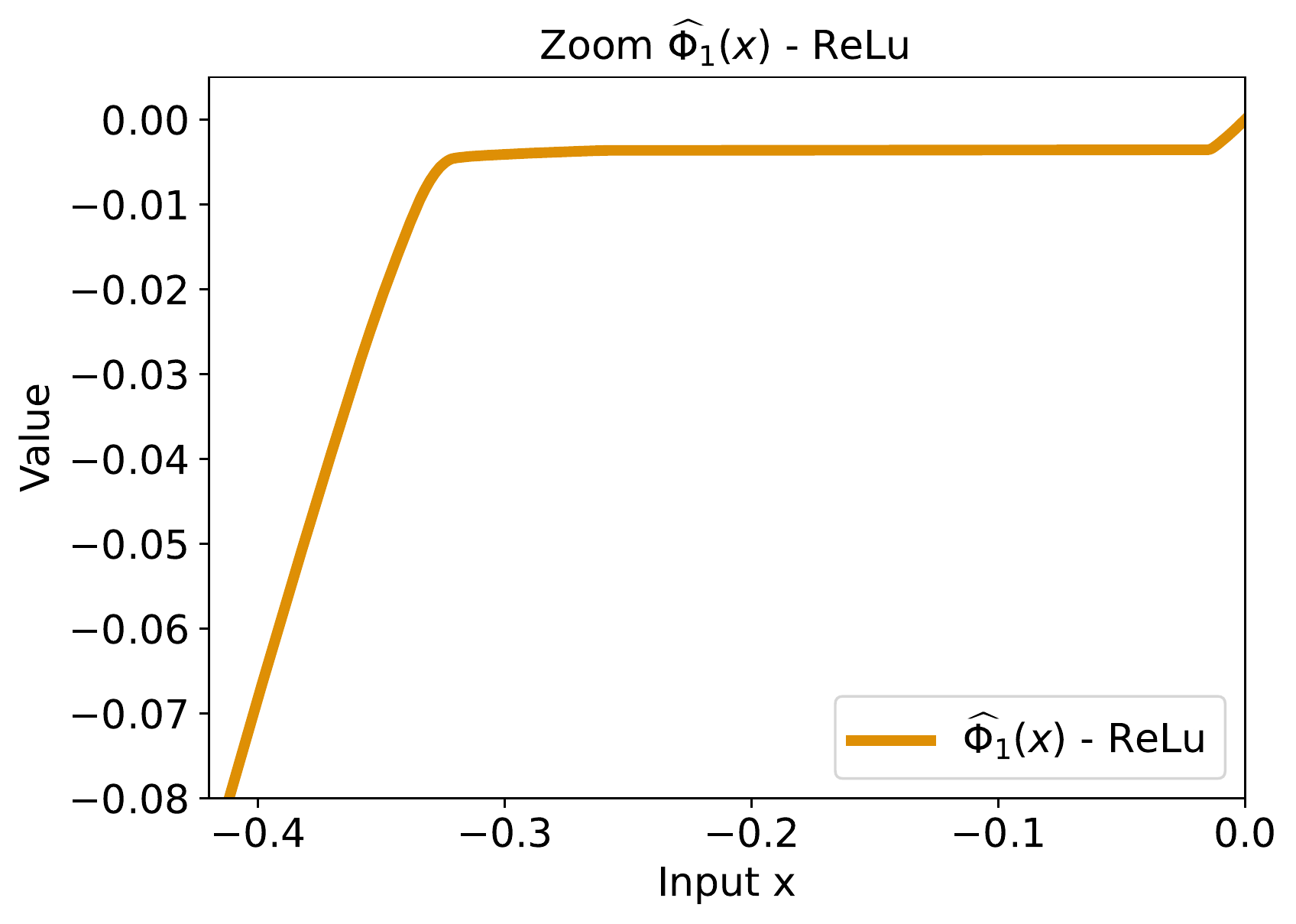}}\label{fig:Loss_Std_ESESEntr_Beta_Big_b}}
			\caption{Heterogeneous Measures - Case Eq. \eqref{Entrconmixedes} - Beta distribution - Convergence Analysis and optimal ReLu. 
				\label{fig:Loss_Std_ESESEntr_Beta_Big}}
		\end{minipage}
	\end{center}
\end{figure}

In Figure \ref{fig:Phi_ESESEntr_Beta_Big} we plot the  average predicted $\hat{\Phi}_1$ and $\hat{\Phi}_2$ for the beta distribution, for the three different activation functions. As in the cases in Section \ref{secsemiexpl} and in the previous heterogeneous case, we anticipate a non-linear behavior, which translates into linear activated DNNs underperforming significantly. We can observe that the optimal allocations found by ReLu and Tanh-activated DNNs are comparable. In Figure \ref{fig:Loss_Std_ESESEntr_Beta_Big} we isolated the allocation $\hat{\Phi}_1$ found by the ReLu DNN which, as in the previous heterogeneous case of Section \ref{sec:hetero}, is the one that performed best, which is confirmed by Table \ref{table:ESESEntr_full}.

\begin{table}
	\begin{center}
		\begin{tabular}{ |c|c|c| } 
\multicolumn{3}{c}{Case Eq. \eqref{Entrconmixedes} - $\mathcal{U}[-1,1])$  }  \\
   \hline
 & Avg. Loss & Std. Loss  \\
 \hline
Linear & $ 0.2064738 $ & $ \mathbf{8.26712 \cdot 10^{-7}} $\\
\textbf{ReLu}&  $ \mathbf{0.1786352} $ & $ 1.22513 \cdot 10^{-6}$  \\
Tanh &  $ 0.1807351 $ & $ 4.32114 \cdot 10^{-5}$ \\
\hline
\multicolumn{3}{c}{}\\
\multicolumn{3}{c}{Case Eq. \eqref{Entrconmixedes} - $\mathcal{N}(0,1)$  }  \\
   \hline
 & Avg. Loss & Std. Loss  \\
 \hline
Linear &$ 0.3662076 $ & $ \mathbf{9.0044 \cdot 10^{-7}} $ \\
\textbf{ReLu} &  $ \mathbf{0.32997002} $ & $ 2.92006 \cdot 10^{-5}$ \\
Tanh &  $ 0.3324212 $ & $ 1.65775 \cdot 10^{-4}$   \\
\hline
\multicolumn{3}{c}{}\\
\multicolumn{3}{c}{Case Eq. \eqref{Entrconmixedes} - $-Beta(2,5)$}  \\
   \hline
 & Avg. Loss & Std. Loss  \\
 \hline
Linear & $ 0.3238797 $ & $ \mathbf{1.00014 \cdot 10^{-6}} $\\
\textbf{ReLu} &  $ \mathbf{0.3132883} $ & $ 1.25900 \cdot 10^{-6}$ \\
Tanh &  $ 0.3147001 $ & $ 1.70660 \cdot 10^{-4}$   \\
\hline
\end{tabular}
\caption{Average Loss of the achieved training losses together with its standard deviation.}
\label{table:ESESEntr_full}
\end{center}
\end{table}

All networks exhibit stable convergence. Still, as expected, the linear-activated DNN achieves a far larger loss level. The Tanh NN underperforms with respect to the ReLu one, yet still provides comparable performances.

\bibliographystyle{BibliographyStyle}
\bibliography{Bibliography}

\begin{appendices}
\section{Implementation details and additional experimental results}
All code is implemented in Python and the Deep Learning library used is PyTorch. In each experiment, the dataset $\Tilde{X}$ is of size $N=100000$, while the \texttt{batch size} is $b=1000$. All the neural networks have $3$ hidden layers of $100$ neurons each and have been optimized with \texttt{Adam}. More precisely, the \texttt{learning rate} is $10^{-6}$ while all other settings of Adam are as per default setting. \label{OnSimplicity} We remind that it is a very well-known result of convex optimization that the learning rate has to be smaller than twice the inverse of the largest eigenvalue of the loss function for Gradient Descent to converge. In practice, this is a valuable indication also in nonconvex optimization. Even if our choice for the learning rate might seem unusual, such a low value was necessary for our experiments, as we observed that higher ones would lead to instability in the optimization process. This is oftentimes an indication that the optimization problem at hand is rather nonlinear and the loss landscape is irregular, together with its derivatives. To make the convergence even more stable, we used the \textit{ReduceLROnPlateau} scheduler for the learning rate, with \texttt{patience} equal to $1000$ and \texttt{threshold} equal to $10^{-6}$, while all other parameters are as per default specification. Finally, all experiments have been run for a number of \texttt{epochs} equal to $300$, apart from those for the Distortion Measures where the number of \texttt{epochs} is $200$. The optimal hyperparameters are the result of fine-tuning via extensive grid search.

We finally complete the exposition of the numerical results for the entropic risk measure and expected shortfall experiments. Table \ref{table:Entr_full} and \ref{table:ES_full} contain the average relative error with respect to the theoretical infimum, together with its standard deviation, and the $L^2$ error of $\hat{\Phi}_1$ with respect to the theoretical $\hphi$ for all distributions and activation functions.

In our experiments, we observe that both ReLu and Tanh activation functions performed well in all cases, even when the solution was known to be linear. ReLu seemed to perform better in most of the cases. This is due to the fact that in some cases the semi-explicit solution has a piecewise linear behavior.

\subsection{Possible Enhancements}\label{sub:PossEnh}

The deep learning literature offers a variety of architectural and methodological enhancements that could be used to further push the results that we obtained. 

\noindent One could include several other activation functions, such as GELU \citep{hendrycks2016gaussian} or ELU which obtained recent success in NLP \citep{brown2020language} and Image Classification \citep{clevert2015fast}, respectively. Similarly, one could try different optimizers which may converge to more stable regions of the loss landscape. For example, recent optimizers that found great success in NLP and Computer Vision are SAM and its variants. As detailed in \citep{wen2022does} and \citep{compagnoni2023sde}, this class of optimizers drives the dynamics towards flatter regions of the landscapes which result in provenly more stable DNNs. Other possibilities include standard techniques such as Batch Normalization and Residual Connections which are proven to stabilize the optimization process.

\label{onmonot} Finally, since the functions we are learning are monotonic, an interesting approach, suggested by an anonymous referee, would be to enforce the monotonicity of the approximating functions. This could be attained by leveraging specific network structures such as in \citep{daniels2010monotone} or suitable penalization terms \citep{liu2020certified}. While all our experiments reached convergence without the need of imposing monotonicity, this might be necessary in other cases where convergence is more elusive.
As a side note, we remark that not enforcing a priori monotonicity allows for a further sanity check in the experiments, as we can check if the monotone behavior of the optima is learned without any external enforcement.

It is worth noting that, for all the architectural changes that would alter the DNNs, one should of course provide the proof of suitable versions of the Theorems \ref{thm:approx1}, \ref{thm:approx2}  and Theorem \ref{thm:main_conv} for this very specific class of NNs.
Since our experiments already achieved satisfactory results, there is no compelling reason to do so at the moment, and we leave these for future research.

\begin{table}
	\begin{center}
		\begin{tabular}{ |c|c|c|c| } 
\multicolumn{4}{c}{Entropic case - $\mathcal{U}[-1,1]$  - Infimum = $0.03423$ }  \\
   \hline
 & Avg. Rel. Error & Std. Rel. Error & Avg. $L^{2}$ Error $\widehat{\Phi}_1$ \\
 \hline
Linear & $\mathbf{\% 2.466 \cdot 10^{-4}} $ & $\mathbf{\% 2.888 \cdot 10^{-4}}$ & $\mathbf{1.207 \cdot 10^{-8}}$\\
ReLu & $\% 9.067 \cdot 10^{-4}$ & $\% 4.682 \cdot 10^{-4}$ & $1.802\cdot 10^{-4}$ \\
Tanh & $\% 1.135\cdot 10^{-3}$ & $\% 3.498 \cdot 10^{-4}$ & $1.077 \cdot 10^{-4}$  \\
\hline
\multicolumn{4}{c}{}\\
\multicolumn{4}{c}{Entropic case - $\mathcal{N}(0,1)$  - Infimum = $0.09656$ }  \\
   \hline
 & Avg. Rel. Error & Std. Rel. Error & Avg. $L^{2}$ Error $\widehat{\Phi}_1$ \\
 \hline
Linear & $\% 1.800 \cdot 10^{-5} $ & $\% 1.172 \cdot 10^{-5}$ & $\mathbf{1.788 \cdot 10^{-7}}$\\
ReLu & $\% \mathbf{5.143 \cdot 10^{-6}}$ & $\% 6.341 \cdot 10^{-5}$ & $4.096 \cdot 10^{-5}$ \\
Tanh & $\% 5.955\cdot 10^{-6}$ & $\% \mathbf{1.814 \cdot 10^{-6}}$ & $5.813 \cdot 10^{-5}$ \\
\hline
\multicolumn{4}{c}{}\\
\multicolumn{4}{c}{Entropic case  - $Beta(2,5)$  - Infimum = $0.2876$  }  \\
   \hline
 & Avg. Rel. Error & Std. Rel. Error & Avg. $L^{2}$ Error $\widehat{\Phi}_1$ \\
 \hline
Linear & $\% 8.979 \cdot 10^{-5} $ & $\mathbf{\% 2.128 \cdot 10^{-5}}$ & $\mathbf{3.759 \cdot 10^{-8}}$\\
ReLu & $\mathbf{\% 7.943 \cdot 10^{-5}}$ & $\% 5.504 \cdot 10^{-5}$ & $3.146\cdot 10^{-4}$ \\
Tanh & $\% 1.001\cdot 10^{-3}$ & $\% 3.202 \cdot 10^{-5}$ & $1.299 \cdot 10^{-4}$\\
\hline
\end{tabular}
\caption{Average relative errors between the losses and the theoretical infimum, their standard deviation, and the Average $L^{2}$ Error between $\widehat{\Phi}_1$ and $\Phi_1$.}
\label{table:Entr_full}
\end{center}
\end{table}

\begin{table}
	\begin{center}
		\begin{tabular}{ |c|c|c|c| } 
\multicolumn{4}{c}{Expected Shortfall  - $\mathcal{U}[-1,1])$  - Infimum = $0.2006$ }  \\
   \hline
 & Avg. Rel. Error & Std. Rel. Error & Avg. $L^{2}$ Error $\widehat{\Phi}_1$ \\
 \hline
Linear & $\mathbf{\% 2.846 \cdot 10^{-4}} $ & $\mathbf{\% 2.270 \cdot 10^{-3}}$ & $\mathbf{9.701 \cdot 10^{-6}}$\\
ReLu & $\% 1.379 \cdot 10^{-2} $ & $\% 8.171 \cdot 10^{-3}$ & $1.817 \cdot 10^{-2}$ \\
Tanh & $\% 2.742 \cdot 10^{-2}$ & $\% 3.301 \cdot 10^{-3}$ & $2.521 \cdot 10^{-3}$ \\
\hline
\multicolumn{4}{c}{}\\
\multicolumn{4}{c}{Expected Shortfall - $\mathcal{N}(0,1)$  - Infimum = $0.3459$ }  \\
   \hline
 & Avg. Rel. Error & Std. Rel. Error & Avg. $L^{2}$ Error $\widehat{\Phi}_1$ \\
 \hline
Linear & $\mathbf{\% 2.961 \cdot 10^{-3}} $ & $\mathbf{\% 1.070 \cdot 10^{-3}}$ & $\mathbf{3.050 \cdot 10^{-5}}$\\
ReLu & $\% 1.910 \cdot 10^{-2}$ & $\% 1.036 \cdot 10^{-2}$ & $3.325\cdot 10^{-2}$  \\
Tanh & $\% 1.530\cdot 10^{-1}$ & $\% 1.867 \cdot 10^{-2}$ & $2.511 \cdot 10^{-2}$  \\
\hline
\multicolumn{4}{c}{}\\
\multicolumn{4}{c}{Expected Shortfall - $Beta(2,5)$  - Infimum = $0.3343$ }  \\
   \hline
 & Avg. Rel. Error & Std. Rel. Error & Avg. $L^{2}$ Error $\widehat{\Phi}_1$ \\
 \hline
Linear & $\mathbf{\% 6.031 \cdot 10^{-4}} $ & $\mathbf{\% 5.253 \cdot 10^{-4}}$ & $\mathbf{1.813 \cdot 10^{-5}}$\\
ReLu  & $\% 1.455 \cdot 10^{-3} $ & $\% 5.449 \cdot 10^{-4}$ & $3.969 \cdot 10^{-2}$ \\
Tanh & $\% 1.017\cdot 10^{-1}$ & $\% 1.506 \cdot 10^{-2}$ & $2.284 \cdot 10^{-3}$\\
\hline
\end{tabular}
\caption{Average relative errors between the losses and the theoretical infimum, their standard deviation, and the Average $L^{2}$ Error between $\widehat{\Phi}_1$ and $\Phi_1$.}
\label{table:ES_full}
\end{center}
\end{table}

\section{Modeling Alternatives} \label{app:AltMod}
As suggested by an anonymous referee, there might be other possible ways to successfully model the functions $f$ and $\ide - f$, for example, using a basis-based approach, such as Random Feature Models, \citep{rahimi2008uniform} or using Kernel functions \citep{scholkopf2018learning}.

In the basis-based approach, it is required to fix (or randomly generate) a number of representations of the input and then to linearly combine them to fit the output via a linear layer. These techniques have proven to be effective and computationally cheap in many fields \citep{NIPS2007_013a006f}. However, the key to their success is a careful design and selection of the (possibly random) features, an operation which is not always straightforward \citep{Compagnoni2022OnTE}. Much differently, DNNs are able to learn and adapt the features during the optimization procedure.

The second approach is based on Reproducing Kernel Hilbert Space (RKHS), also known in the Machine Learning community as \textit{kernel methods}. This is a very powerful set of techniques that maps the input data into a higher (possibly infinite) dimensional space, in which it is easier to separate data points respect to their native space. These methods found success in many applications \citep{scholkopf2018learning} such as in Classification, Signal Detection \citep{kailath1975rkhs}, and Function Emulation \citep{scholkopf2015computing}.
However, we find that the \textit{kernel trick} \citep{theodoridis2006pattern} at the basis of these methods does not allow us to find a closed-form solution for our problem. Therefore, while this would allow us to face a convex optimization problem, we would still have to rely on an optimizer such as SGD (or Adam) to actually find the unique solution. In this regard, we recall that using RKHS requires calculating the Gramian matrix, which has a complexity of $\mathcal{O}(N^2)$, where $N$ is the number of data points. Therefore, even just evaluating the loss function in Eq. \eqref{eq:loss} has a complexity of $\mathcal{O}(N^2)$, for each training epoch. This cost is additional to the computation of gradients and the update of the parameters in the optimization step, therefore, we expect a much higher computational cost and less scalability of RKHS-based techniques compared to that of DNN. From a theoretical point of view, consistency results for RKHS-based techniques in the literature are only available for the Supervised Learning case and it is not clear if they would be easily adapted to our Unsupervised Learning setting.

To conclude, while many alternatives are present, many of them present criticalities such as higher computational cost and design challenges, that DNNs easily avoid.

\end{appendices}

\vspace{0.5cm}

\textbf{Disclosure statement}: The authors report there are no competing interests to declare. 

\vspace{0.5cm}

\textbf{Acknowledgements}: The authors thank two anonymous referees for precious comments, and F.-B. Liebrich for addressing them to the reference \citep{Shapiro13} and for pointing out the delicate point of the standardness requirements on the underlying probability space.

 \end{document}